%% file: main.tex
\documentclass[conference]{IEEEtran}
\IEEEoverridecommandlockouts
\usepackage{cite}
\usepackage{amsmath,amssymb,amsfonts}
\usepackage{algorithmic}
\usepackage{graphicx}
\usepackage{textcomp}
\usepackage{color}
\usepackage{xcolor}
\usepackage{array}
\usepackage{booktabs}
\usepackage{makecell}
\usepackage{xspace}
\usepackage{mathtools}
\usepackage[linesnumbered,ruled,vlined]{algorithm2e}
\usepackage{amsthm}
\usepackage{paralist}
\usepackage{graphicx}
\usepackage{verbatim}
\usepackage{enumitem}
\usepackage[center]{subfigure}
\usepackage{multirow}
\usepackage{bm}
\usepackage{float}
\usepackage{url}

\def\BibTeX{{\rm B\kern-.05em{\sc i\kern-.025em b}\kern-.08em
    T\kern-.1667em\lower.7ex\hbox{E}\kern-.125emX}}

\input{dfn}

\begin{document}

\title{Fast Searching The Densest Subgraph And Decomposition With Local Optimality}

\author{
Yugao Zhu$^1$, Shenghua Liu$^1$, Wenjie Feng$^{2}$, 
Xueqi Cheng$^1$ 
\vspace{1.8mm}
\\
\fontsize{10}{10}\selectfont\itshape
$^1$Institute of Computing Technology, Chinese Academic of Sciences, Beijing, China\\
$^2$Institute of Data Science, National University of Singapore, Singapore\\
\fontsize{8}{8}\selectfont\ttfamily\upshape
zhuyugao22@mails.ucas.ac.cn, liushenghua@ict.ac.cn, wenjie.feng@nus.edu.sg, cxq@ict.ac.cn

}

\maketitle

\begin{abstract}
    \input{000abstract}
\end{abstract}

\begin{IEEEkeywords}
Densest subgraph problem, Data mining, Algorithm design, Optimization.
\end{IEEEkeywords}

\section{Introduction}
    \label{sec:intro}
    \input{010intro}

\section{preliminaries}
    \label{sec:pre}
    \input{030prelims}

\section{methodology}
    \label{sec:method}
    \input{040method}

\section{experiment}
    \label{sec:exp}
    \input{050exps}

\section{Related work}
    \label{sec:related}
    \input{020related}

\section{Conclusion and Future Outlook}
    \label{sec:con}
    \input{060conclusion}

\section{Acknowledgments}
    \input{070ack}
    

\newpage
\section{Full version}
    \input{080sup}

\end{document}

%% file: dfn.tex
\newcommand{\graph}{\mathcal{G}\xspace}
\newcommand{\edges}{\mathcal{E}\xspace}
\newcommand{\nodes}{\mathcal{V}\xspace}
\newcommand{\weights}{\mathcal{W}\xspace}
\newcommand{\subnode}{\mathcal{S}\xspace}
\newcommand{\setndeg}[2]{\mathrm{d}_{#1}(#2)\xspace}
\newcommand{\numR}{\mathbb{R}}
\newcommand{\optset}{\mathcal{S}^{*}\xspace}
\definecolor{babyblueeyes}{rgb}{0.19, 0.55, 0.91}
\newcommand{\linecomment}[1]{\textcolor{babyblueeyes}{\textit{$\triangleright$ #1}}}

\newcommand{\first}[1]{\textbf{\underline{#1}}}
\newcommand{\second}[1]{\underline{#1}}
\newcommand{\hide}[1]{}

\newcommand{\atn}[1]{\textcolor{red}{#1}}

\newtheorem{problem}{\textbf{Problem}}
\newtheorem{definition}{\textbf{Definition}}
\newtheorem{theorem}{\textbf{Theorem}}
\newtheorem{lemma}{\textbf{Lemma}}
\newtheorem{property}{\textbf{Property}}
\newtheorem{complexity}{\textbf{Complexity}}
\newtheorem{corollary}{\textbf{Corollary}}

\DeclareMathOperator*{\argmax}{arg\,max}
\DeclareMathOperator*{\argmin}{arg\,min}
\DeclareMathOperator*{\argsort}{arg\,sort}

\newcommand{\method}{\textsc{Lowd}\xspace}

%% file: 000abstract.tex
Densest Subgraph Problem (DSP) is an important primitive problem with a wide range of applications, including fraud detection, community detection, and DNA motif discovery. Edge-based density is one of the most common metrics in DSP. Although a maximum flow algorithm can exactly solve it in polynomial time, the increasing amount of data and the high complexity of algorithms motivate scientists to find approximation algorithms. Among these, its duality of linear programming derives several iterative algorithms including Greedy++, Frank-Wolfe, and FISTA which redistribute edge weights to find the densest subgraph, however, these iterative algorithms vibrate around the optimal solution, which is not satisfactory for fast convergence. We propose our main algorithm Locally Optimal Weight Distribution (\method) to distribute the remaining edge weights in a locally optimal operation to converge to the optimal solution monotonically. Theoretically, we show that it will reach the optimal state of specific quadratic programming, which is called locally-dense decomposition. Besides, we show that it is not necessary to consider most of the edges in the original graph. Therefore, we develop a pruning algorithm using a modified Counting Sort to prune graphs by removing unnecessary edges and nodes, and then we can search the densest subgraph in a much smaller graph.

%% file: 010intro.tex

Finding subgraphs with the highest average degrees in large networks is an important primitive problem in data mining, 
and has been applied to different areas including social networks~\cite{shen2010spectral}, biological analysis~\cite{wong2018sdregion}, traffic pattern mining~\cite{liu2019coupled}, 
and graph database~\cite{cohen2003reachability, jin20093}. 
The maximum-flow-based algorithm~\cite{goldberg1984finding} can exactly solve DSP by utilizing the binary search in polynomial time,
which is ill-suited for large graphs due to the prohibitive cost.
The greedy-peeling algorithm (`Greedy' as an abbreviation) proposed by Charikar~\cite{charikar2000greedy} can provide a $1/2$-approximation guaranteed solution in linear time by greedily deleting the nodes with the minimum degree.
There are also various variants for both the exact and approximation algorithms applied to different scenarios.

\begin{figure*}[htbp]
    \centering
    \includegraphics[scale=0.26]{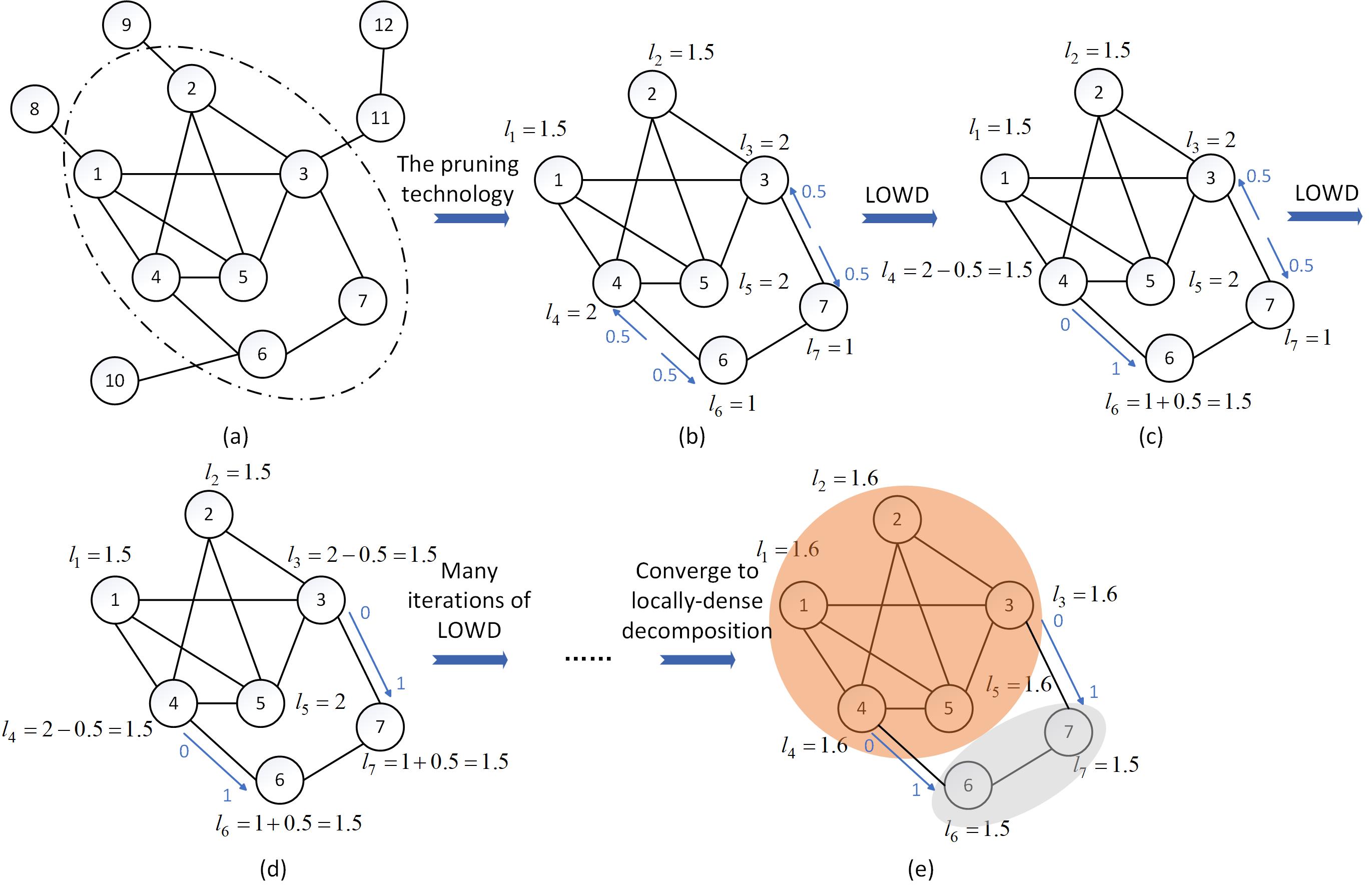}
    \vspace{-0.1in}
    \caption{An example of our methods. $(a)\to (b)$: Remove some nodes using our pruning technique. $(b)$: Initialize node loads by distributing edge weight equally to two endpoints. $(b)\to (c)$: distribute the weight of the edge connecting nodes 4 and 6 to node 6 in one-way. $(c)\to (d)$: distribute the weight of the edge connecting nodes 3 and 7 to node 7 in one-way.  $(d)\to (e)$: After many iterations of \emph{\method}, the graph converges into the locally dense decomposition, that is, loads of nodes 1, 2, 3, 4, and 5 gradually become equal through the redistribution of edge loads and form a nodeset. Nodes 6 and 7 form a nodeset, and edges between two different nodesets are undirectionally distributed from the nodeset with a higher load to the nodeset with a lower node load. And deleting the node with the lowest load one by one can find the densest subgraph, that is, the subgraph consisting of nodes 1, 2, 3, 4, and 5. The order in which edges are selected during edge weight redistribution does not affect the result.}
    \label{fig:example} 
\end{figure*}

In addition, to detect the densest subgraph more efficiently, in recent years, many DSP algorithms and discoveries towards its dual linear programming (LP) are developed, which was first introduced in \cite{charikar2000greedy}. 
E.g., Danisch et al. \cite{danisch2017large} redistribute edge weights based on the Frank-Wolfe algorithm, 
\cite{sawlani2020near} relaxes the constraints of LP dual and proposes an algorithm for dynamic graphs, 
Greedy++ was proposed in \cite{boob2020flowless} based on the multiplicative weights update (MWU) framework, which iteratively implements Greedy to search the densest subgraph. 
In \cite{chekuri2022densest}, it gives a theoretical guarantee to the relationship between iteration count and approximation ratio of Greedy++, and popularizes Greedy++ to super-modular function $f(S)$, i.e., 
the definition of density is $\frac{f(S)}{|S|}$. Besides, \cite{harb2022faster} proposes FISTA algorithm with a faster convergence speed than Greedy++.
Frank-Wolfe, Greedy++, and FISTA are all developed based on classical methods towards convex optimization, like gradient descent and MWU. 
However, these iterative methods will vibrate around the optimal solution, which is time-consuming. 

In this paper, We elaborately design a novel algorithm called \emph{locally optimal weight distribution} (\method) algorithm, 
based on the dual problem of LP for DSP using a locally optimal operation with monotonic convergence to the optimal value, 
it detects the densest subgraph more efficiently than previous studies.
Besides, we develop a pruning technology to remove most nodes that don't belong to the densest subgraph before using \emph{\method} to detect the densest subgraph. Some pruning technique is used in \cite{fang2019efficient} but we use a modified Counting Sort to make our pruning comply with linear time complexity on unweight graphs, we also prove that our pruning technique is a subprocess of Greedy. In addition, we theoretically prove that \method converges to the locally-dense decomposition solution, which is a well-studied problem \cite{khuller2009finding,tatti2015density,danisch2017large,harb2022faster,ma2022finding,tatti2019density}. The locally-dense decomposition consists of a set of subgraphs with nested structures and densities. Such decomposition can be derived by various iterative methods, including Frank-Wolfe, Greedy++, and FISTA.

Moreover, we comprehensively verify the performance of \method over 26 real-world networks.
The experimental results show that \method has 
a better convergence rate than other iterative algorithms, i.e., Frank-Wolfe, Greedy++, and FISTA, when detecting the densest subgraph and locally-dense decomposition.


In Figure \ref{fig:example}, we illustrate an example of \emph{\method} algorithm, and the final locally-dense decomposition. First, we use the pruning technique to delete unnecessary nodes in $(a)$. In the initial state of \method, we distribute each edge weight equally to endpoints, accumulating it as the node loads on the endpoints. Then, \method redistributes edge weights to minimize the difference between two endpoints' loads, for example, in $(b)$ and $(c)$, node 4 with load 2 transits edge weight to node 6 with load 1, and their loads both become 1.5, a similar scenario happens on nodes 5 and 7. Therefore, in order to make node loads more balanced, a node with a higher load should transit edge weight to its neighbors with lower loads. From a global perspective, edge weight should be redistributed from nodes with higher loads to nodes with lower loads in iterations, and the densest subgraph can be easily discovered by deleting the node with the lowest load one by one. And the graph finally converges to the locally-dense decomposition like $(e)$.

In summary, our main contributions include as follows:
\begin{itemize}
    \item[$\bullet$] We propose an iterative mining method called \emph{\method} for DSP according to its LP dual and prove it will converge into the locally-dense decomposition, which is a variant of DSP. We also exhibit how \emph{\method} solves both problems in a locally optimal operation iteratively.
    \item[$\bullet$] We develop a pruning technique to lock the densest subgraph into a graph with a much smaller size using a modified Counting Sort and prove that it is a subprocess of Greedy search.
    \item[$\bullet$] We did sufficient experiments on 26 real-world datasets in various fields, with sizes up to hundreds of millions of edges. Our result shows that \emph{\method} can detect the densest subgraph much faster than other baselines and it can converge to the locally-dense decomposition more efficiently, while the pruning technique as a pre-process can effectively reduce the computation. \end{itemize}

\textbf{Organization.} We organize the rest paper as below. In section \ref{sec:pre} we present two problems formally: the densest subgraph problem and the locally-dense decomposition problem. In section \ref{sec:method} we proposes an iterative method to deal with the linear programming of DSP and quadratic programming of locally-dense decomposition, and develop a pruning technique to lock the densest subgraph into a much smaller graph. We exhibit their efficient performance based on various experiments in section \ref{sec:exp}. In section \ref{sec:related} we review the related work about DSP and its variants. We summarize our work and make a prospect for future work In section \ref{sec:con}.

Due to space limits, part of the theorems, lemmas, and experimental results are given in the appendix.


%% file: 030prelims.tex
\begin{table}[t]
    \centering
    \caption{\textnormal{Symbols and Definitions.}}
    \label{tab:notations}
    \vspace{-0.1in}
    \begin{tabular}{c|c} \toprule
        \textbf{Symbol} & \textbf{Definition and Description}\\ \midrule
        $\graph(\nodes,\edges,\weights)$ & \makecell{Graph $\graph$ with nodeset $\nodes$,\\ Edgeset $\edges$ and edge weights $\weights$} \\
        $M$,$N$         & Number of edges and nodes of a graph \\
        $\subnode$          & Subset of nodes, i.e., $\subnode \subseteq \nodes$ \\
        $\subnode^{*}$ & The nodeset of the densest subgraph\\
        $\edges(\subnode),\weights(\subnode)$& the edgeset and total weights induced by $\subnode$\\
        $\rho \left ( \subnode \right )$ & The density of the nodeset $\subnode$\\
        $\rho^{*}$ & The density of the densest subgraph\\
        $\setndeg{\subnode}{u}$ & The degree of node $u$ in $\subnode$ \\
        $l_u, w_e$            & Load of the node $u$ and Weight of the edge $e$ \\
        $f_{e}(u)$            & The weight distributed to node $v$ from edge $e$ \\
        \bottomrule
    \end{tabular}
\end{table}

In this section, we formally define the densest subgraph detection and 
locally-dense decomposition problem that we focus on. 
Table~\ref{tab:notations} summarizes the main symbols used in the paper.

Let $\graph = (\nodes, \edges, \weights)$ be an undirected graph 
with $N = |\nodes|$ vertices and $M = |\edges|$ edges.
For any $e \in \edges \subseteq \nodes \times \nodes$, 
its weight is $w_{e} \in \weights$ with $w_{e} \in \numR_{+}$ 
and $w_{e} = 1$ for the weighted and unweighted graph, respectively.
Given a node subset $\subnode \subseteq \nodes$, 
$\edges(\subnode)$ denotes the set of edges and 
$\weights(\subnode)$ denotes the total weights induced by $\subnode$,
and $\setndeg{\subnode}{u}$ is the degree of $u$ in the induced subgraph, 
i.e., the total weights of edges connected to $u$ within the set $\subnode$. Bold letters are used to represent vectors.

Given the nodeset $\subnode$, 
the \emph{edge density} of the subgraph $\graph(\subnode)$ is defined as 
\begin{equation*}
    \rho(\subnode) \coloneqq \frac{\weights(\subnode)}{|\subnode|} 
= \frac{\sum_{u \in \subnode}\setndeg{\subnode}{u}}{2 \cdot |\subnode|}.
\end{equation*}
Accordingly, we present the formal definition of the densest subgraph problem as below:

\begin{problem}[Densest Subgraph Problem (DSP)]
    \label{prob:dsp}
    Given an undirected graph $\graph = (\nodes, \edges, \weights)$, 
    find the subset of nodes $\optset$ such that 
    $\optset = \argmax_{\subnode \subseteq \nodes} \rho(\subnode)$.
\end{problem}

Another variant of DSP is more general, which is called locally-dense decomposition.

\begin{definition}[Locally-Dense Decomposition (LDD) \cite{tatti2015density,tatti2019density}]
    Given an undirected graph $\graph=(\nodes, \edges, \weights)$, 
    it has a nested decomposition consisting of a sequence 
    $\emptyset=B_0 \subsetneqq B _1 \subsetneqq \ldots \subsetneqq B_k = \nodes$. 
    We define $B_i$ as the maximal densest subgraph properly containing $B_{i-1}$, that is,
    $$ B_i = \argmax_{{S \supsetneqq B_{i-1}}}{\frac{\weights(S)-\weights(B_{i-1})}{| S \setminus B_{i-1}|}} $$
\end{definition}


As we can see, DSP is a sub-problem of locally-dense decomposition, because $B_1$ is the maximal densest subgraph of the whole graph $\graph$, corresponding to the target in Problem~\ref{prob:dsp}. LDD is also closely related to convergence analysis for iterative update methods towards DSP~\cite{danisch2017large, boob2020flowless, harb2022faster}.
Formally, the locally-dense decomposition problem over $\graph$ is formulated as:

\begin{problem}[Locally-dense decomposition             Problem\cite{tatti2015density,tatti2019density,danisch2017large}]
    \label{prob:locally-dense}
    Given an undirected graph $\graph=(\nodes, \edges, \weights)$, 
    find its locally-dense decomposition 
    $\emptyset=B_0 \subsetneqq B _1 \subsetneqq \ldots  \subsetneqq B_k=\nodes$.
\end{problem}

Based on the above problem definition, 
we summarize some important properties of LDD.

\begin{property}[\cite{danisch2017large}]
    Given an undirected graph $\graph=(\nodes, \edges, \weights)$, its locally-dense decomposition is unique. 
    \label{prop:uniquity1}
\end{property}

\begin{property}[\cite{harb2022faster}]
    And for any $u \in B_i$, let 
    \begin{equation*}        \lambda_u=\lambda_i\coloneqq{\frac{\weights(\edges(W))-\weights(\edges(B_{i-1}))}{|W\setminus B_{i-1}|}},
    \end{equation*}
    there is a unique optimal $\bm{\ell}^*$ so that ${\ell_v}^*=\lambda_v$ for each node $v$. And $\lambda_1>\lambda_2>...>\lambda_k$.
    \label{prop:uniquity2}
\end{property}

\begin{property}[\cite{danisch2017large,ma2022finding}]
    For an optimal solution $(\bm{f}^*,\bm{\ell}^*)$, if there is an edge $e=(v_1,v_2)$ with ${\ell_{v_1}}^*>{\ell_{v_2}}^*$, then ${f_e}^*(v_1)=0$.
    \label{prop:one-way}
\end{property}

Properties \ref{prop:uniquity1} and \ref{prop:uniquity2} imply that in the unique locally-dense decomposition, each node in set $B_i\setminus B_{i-1}$ has the same load $\lambda_i$, and the nodeset in the inner nodeset in the inner layer is denser than one in the outer layer. Property \ref{prop:one-way} means the weight of edge e connecting two layers should be only distributed to the node with a lower load in one-way. 

In section \ref{sec:method} we will introduce an algorithm which provides an iterative operation to solve both DSP and locally-dense decomposition problem.

%% file: 040method.tex
\begin{algorithm}[bt]
    \caption{\method: Locally optimal weight distribution for densest subgraph detection}
    \label{alg:lowd}
    \KwIn{Undirected graph $\graph$, iteration count $T$.}
    \KwOut{An approximately densest subgraph of $\graph$.}
        \For{$e = (u, v) \in \edges$}{
            $f_e(u) = f_e(v) = \frac{w_e}{2}$;
            \hfill \linecomment{initialize edge weight}
        }
        \For{$u \in \nodes$}{
            $\ell_u = \sum_{e\in \edges:u\in e}{f_e(u)}$
            \hfill \linecomment{initialize node load}
        }
        \For{$k : 1 \to T$}{
            \For{$e = (u, v) \in \edges$}{ 
                \If{$ \ell_{u} > \ell_{v}$}{
                    \hfill \linecomment{balance nodes loads as much as possible}
                    
                    $d \leftarrow \min \{(\ell_{u} - \ell_{v})/2, \, f_{e}(u) \} $;

                    $\ell_{u} \leftarrow \ell_{u} - d$;
                    
                    $ f_e(u) \leftarrow f_e(u) - d$;  
                    
                    $\ell_{v} \leftarrow \ell_{v} + d$;
                    
                    $f_e(v) \leftarrow f_e(v) + d$;
                } 
                \Else {
                    $d \leftarrow \min \{(\ell_{v} - \ell_{u})/2, \, f_{e}(v) \} $;
                    
                    $\ell_{u} \leftarrow \ell_{u} + d$;
                    
                    $f_e(u) \leftarrow f_e(u)+d$;
                     
                    $\ell_{v} \leftarrow \ell_{v} - d$;
                    
                    $f_e(v) \leftarrow f_e(v)-d$;
                }
            }
        }
        $\optset, \subnode \leftarrow \nodes, \nodes $; 
        \hfill \linecomment{sort $l_u$ in a non-decreasing order}
        
        $V_{s} \leftarrow \argsort_{u \in \nodes} l_u$;
        
        \For{$i :  1 \to N$}{
            $\subnode \leftarrow \subnode \setminus \{\nodes_{s}(i)\}$; \hfill \linecomment{$\nodes_{s}(i):$ the $i$-th element of $\nodes_{s}$}
            
            \If{$\rho(\optset) < \rho(\subnode)$}{
                $\optset \leftarrow \subnode$;
            }   
        }
    \Return{$ \graph(\optset) $.}
\end{algorithm}

In this section, through the lens of the primal-dual formulation of linear programming, 
we illustrate our idea for solving DSP
and propose the \textbf{l}ocally \textbf{o}ptimal \textbf{w}eight \textbf{d}istribution algorithm, \emph{\method}, 
which is a fast iterative approach to searching the densest subgraph according to the LP dual of DSP.
We also theoretically prove that \emph{\method} makes the graph converge into locally-dense decomposition.
To search the densest subgraph more efficiently, we propose a pruning algorithm, rendering 
the necessary condition for DSP and corresponding to a subprocess of Greedy. 

\subsection{Locally Optimal Weight Distribution Algorithm}

Firstly we present the pseudo code of \emph{\method} in algorithm \ref{alg:lowd} and explain it. In algorithm \ref{alg:lowd}, given an undirected graph $\graph$ and iteration count T, we distribute each edge $e$'s weight to its two endpoints $u$ and $v$ and use vector $\bm{f}$ to describe it. Specifically, $f_e(u)$ is the weight distributed to node $u$ from edge $e$. Besides, we use vector $\bm{\ell}$ to remark the load of each node received from all corresponding edges, specifically, $\ell_u=\sum_{e\in \edges:u\in e}{f_e(u)}$. In lines 1-4, we distribute each edge equally to two endpoints as the initial state and
accordingly calculate node loads. Its time complexity is $O(M)$. Lines 5-14 mean that it redistributes the weight of edge $e=(u,v)$ to minimize the difference of $\ell_u$ and $\ell_v$ for every $e\in \edges$ in iterations. Its time complexity is $O(MT)$. In lines 7-10 we distribute more edge weight from $u$ to $v$, and update step $d$ is $(\ell_{u}-\ell_{v})/2$ if it will not make any $f_{e}(u)<0$, otherwise $d$ is $f_{e}(u)$. In lines 11-14 the circumstance is the opposite. In line 15 we set the whole nodeset as the initial nodeset $S$. In lines 16-20 we sort S according to node loads and delete the node with the lowest weight one by one to get a subgraph with high density, its time complexity is $O(M+NlogN)$.

\begin{complexity}
    The time complexity of \emph{\method} is $O(MT+M+NlogN)$.
\end{complexity}

\textbf{Remark.} Although \emph{\method} can detect the densest subgraph with enough iterations, but we cannot determine iteration count required. Therefore, our \emph{\method} belongs to the approximation algorithm in strict terms because actually we just set the iteration count to get an approximate solution for DSP. The same circumstance happens on other approximation algorithms including Frank-Wolfe in \cite{danisch2017large}, Greedy++ in \cite{boob2020flowless} and FISTA in \cite{harb2022faster}.


\subsection{\emph{\method} converges to DSP solution}

In order to figure out how \emph{\method} searches the densest subgraph iteratively, we first introduce the LP primal-dual of DSP from \cite{charikar2000greedy,boob2020flowless,danisch2017large} as follows. The notation of LP primal-dual is the same as \cite{boob2020flowless}.

\begin{equation}
    \centering
    \label{eq:primal}
    \begin{aligned}
        \textrm{maximize}   \qquad & \sum_{e\in \edges}{w_{e}y_{e}} \\
        \textrm{subject to} \qquad & y_e \le x_u, \qquad \forall e = uv \in \edges \\
                            & y_e \le x_v, \qquad \forall e = uv \in \edges \\
                            & \sum_{v \in \nodes}{x_v} \le 1 \\
                            & y_e \ge 0 \qquad \forall e \in \edges \\
                            & x_v \ge 0 \qquad \forall v \in \nodes \\
    \end{aligned}
\end{equation}

In LP \eqref{eq:primal}, the binary $x_u$ and $y_e$ indicate the contribution to density of the densest subgraph from node $u$ and edge $e$. The maximum of LP \eqref{eq:primal} is $\rho^*$, i.e., the maximum density in DSP. You can set $y_e=\frac{1}{\optset}$ if $e \in \edges(\optset)$ otherwise $y_e=0$ and set $x_u=\frac{1}{\optset}$ if $u \in \optset$ otherwise $x_u=0$, then you will get $\frac{\weights(\optset)}{\optset}$ as the optimal value of $\sum_{e\in \edges}{w_{e}y_{e}}$. Instead of using the primal problem, we resort to the LP dual for DSP to illustrate our motivation:

\begin{equation}
    \centering
    \begin{aligned}
        \textrm{minimize}  \qquad &  D \\
        \textrm{subject to} \qquad & f_e(u)+f_e(v)\ge w_{e}\qquad\forall e=uv\in \edges \\
                        & \ell_v \overset{\text{def}}{=} \sum_{e \ni v}{f_e(v)} \le D \qquad \forall v \in \nodes\\
                        & f_e(u) \ge0\qquad\forall e = uv \in E \\
                        & f_e(v) \ge0\qquad\forall e = uv \in E \\
    \end{aligned}
    \label{eq:dual}
\end{equation}

From strong duality, we know its optimal value is also $\rho^*$. The symbols in LP \eqref{eq:dual} can be interpreted in accordance with the description in \emph{\method}. $f_e(u)$ and $f_e(v)$ should be both positive and the sum is not less than $w_e$ to ensure the edge is distributed thoroughly. Actually, in the LP dual of DSP, we can keep $f_e(u)+f_e(v)=w_{e}$ instead of $f_e(u)+f_e(v)\ge w_{e}$, because the former can keep constraint condition and doesn't increase optimization objective $D$. Consequently, we can get the optimization objective $D$ with tighter constraint in LP \eqref{eq:tighter_dual} as below.

\textbf{note:} During the whole process we set $D =\max_{v\in \nodes} \ell_v$ to minimize it as much as possible.

\begin{equation}
    \centering
    \label{eq:tighter_dual}
    \begin{aligned}
        \textrm{minimize}   \qquad & D \\
        \textrm{subject to} \qquad & f_e(u) + f_e(v) = w_e \qquad \forall e = uv \in \edges \\
                        & \ell_v = \sum_{e\ni v}{f_e(v)} \le D \qquad \forall v \in \nodes\\
                        & f_e(u) \ge 0 \qquad \forall e = uv \in \edges \\
                        & f_e(v) \ge 0 \qquad \forall e = uv \in \edges \\
    \end{aligned}
\end{equation}

Intuitively speaking, what \emph{\method} does is to redistribute each edge weight to minimize the difference between loads of two endpoints for each edge in iterations, i.e., propagate edge weights from nodes with higher loads to nodes with lower loads. And the impact of \emph{\method} on LP \eqref{eq:tighter_dual} is that the node $v$ with the highest node load $l_{v}=D$ will decrease its weight because $v$ transits its edge weights to its neighbors. After experiencing some iterations, the highest node load will decrease gradually to some value, in fact, the value is the minimum of $D$, i.e., $\rho^*$. Formally we propose the following theorem:

\begin{theorem}
    The iterative operation in \emph{\method} can make the sequence $\left \{ D_t \right \}$ converge to $\rho^*$, where $D_t (t>0)$  means the optimization objective $D$ after $t$ iterations of \emph{\method} and $D_{0}$ is the value in the initial state. 
    \label{th:dsp}
\end{theorem}

Next, we will give theoretical proofs for Theorem \ref{th:dsp}, first we need the following classical theorem for convergence.

\begin{theorem}[Monotone Convergence Theorem~\cite{bibby1974axiomatisations}]
If the sequence $\left \{ a_n \right \} $ has an upper bound and it is monotonically non-decreasing (or has a lower bound and it is monotonically non-increasing), then the sequence $\left \{ a_n \right \} $ converges, i.e., a monotonically bounded sequence must have a limit.
\label{th:tmct}
\end{theorem}

Now that $\left\{D_t \right \}$ has a lower bound $\rho^{*}$ according to strong duality, and $\left \{ D_t \right \} $ is monotonically non-increasing under \emph{\method}'s iterative operation (because in the whole process \emph{\method} just balances loads of two endpoints and it will not produce any new node with load higher than the optimization objective $D$). Then $D_t$ is non-increasing and it must converge to some value.

It is helpless for solving DSP if the optimization objective converges to some value which is not the minimum, i.e., $\rho^{*}$. However, we can use the following lemma to help to prove Theorem \ref{th:dsp}.

\begin{lemma}
    for $\forall t \in \mathbb{N}$, there must be $D_{t+N}<D_{t}$ if $D_t \ne \rho^{*}$, where N is the node number of the whole graph.
    \label{lem:decrease}
\end{lemma}

\begin{proof}
    For $\forall t \in \mathbb{N}$, if $D_t>\rho^{*}$, we make an assertion that there must be some nodes $u$ and $v$ and an edge $e=(u,v)$ with $\ell_u=D_t$, $\ell_v<D_t$ and $f_e(u)>0$. Suppose this doesn't stand up, so if we set $A=\left \{u|l_u=D_t \right \}$ and $B=\left \{u|l_u<D_t \right \}$, for $\, \forall u \in A,v \in B$ and $e=(u,v)$, then $f_e(u)=0$. Therefore we conclude that:
    \begin{equation*}
        \footnotesize
        \centering
        \begin{aligned}
            \rho \left ( A \right ) 
            &=\frac{\sum_{e \in \edges(A)} w_e}{|A|}=\frac{\sum_{e=(u,v),\,u,v\in A} f_e(u)+f_e(v)}{|A|}\\
            &=\frac{\sum_{e=(u,v),\,u,v\in A} {(f_e(u)+f_e(v))}+\sum_{e=(u,v),\,u\in A,v\in B}{f_e(u)}}{|A|}\\
            &=\frac{\sum_{u\in A}{\sum_{e\ni u}{f_e(u)}}}{|A|}\\
            &=\frac{\sum_{u\in A}{l_u}}{|A|}=\frac{\sum_{u\in A}{D_t}}{|A|}=D_t>\rho^*\\
        \end{aligned}
    \end{equation*}

    It will produce a subgraph whose density is larger than the densest subgraph. That will lead to a contradiction. 

    Therefore in each iteration, node $u$ will transit edge weight to node $v$ and then there will be $l_u<D_t$ according to our assertion. Notice that in the whole process, \emph{\method} will not produce any new node $v$ with $l_v\ge D_t$ after t iterations. The number of nodes with load $D_t$ will decrease in each iteration until 0. Given that $N$ is the node number of the whole graph, after $N$ iterations, there isn't any node $v$ with $l_v=D_t$, i.e., $D_{t+N}<D_{t}$.
\end{proof}

Combining Lemma \ref{lem:decrease} and Theorem \ref{th:tmct} is not enough to prove our claim because $\left\{D_t \right \}$ may decrease infinitesimally and converge to another value instead of the minimum. However, when dealing with this difficulty, it is useful to combine the proof of Lemma \ref{lem:decrease} in the limit sense.

\begin{proof}[Proof of Theorem \ref{th:dsp}]
    We adopt the proof by contradiction which is similar to the proof of Lemma \ref{lem:decrease}. Suppose that \emph{\method} makes $\left \{ D_t \right \}$ converge to any other value $D$ which $D > \rho^{*}$. We set $A=\left \{u|l_u\to D \right \}$ and $B=\left \{u|l_u\not\to D \right \}$. When $T \to \infty$, $|A|$ will decrease and converge to a fixed number according to Theorem \ref{th:tmct}, and all edges connecting $A$ and $B$ are distributed to $B$ in the limit sense. Then $\rho(A)$ will be tending to $D$, which makes a contradiction because $D > \rho^{*}$. Only when $D = \rho^{*}$, there will be no contradiction.
\end{proof}

The explanation for lines 19-24 in Algorithm \ref{alg:lowd} is closely related to the locally-dense decomposition, which will be proved in the next subsection. In fact, $B_1$ in LDD is the maximal densest subgraph and nodes in $B_1$ will have the max node load $\rho^*$. Therefore in LDD, as long as we delete all the nodes whose weights are not the maximum, the remaining subgraph is the maximal densest subgraph. If \emph{\method} makes the graph converge into LDD, given that the node load in the densest subgraph will be not completely the same after several iterations, it is safe to delete the node with the lowest weight one by one.

A drawback of the iterative algorithm is that we don't know whether we have found the densest subgraph so as to stop iterations. However, on unweighted graphs, if the difference between optimization objective $D$ and the maximum density found by \emph{\method} is less than $\frac{1}{n(n-1)}$, we can confirm that \emph{\method} has found the densest subgraph and it can stop iterations, which is similar with the maximum flow algorithms in \cite{goldberg1984finding}. Given that the edge weights satisfy $\weights \in \numR_{+}$ on weighted graphs, it doesn't work to use the difference $\frac{1}{n(n-1)}$ to determine whether we have found the densest subgraph if it is weighted.


Next, we will explan the relationship between \emph{\method} and LDD.

\subsection{\emph{\method} converges to LDD's solution}
\label{subsecion:ldd}
There are many iterative methods dealing with the LP dual of DSP including Frank-Wolfe in \cite{danisch2017large}, Greedy++ in \cite{boob2020flowless} and FISTA in \cite{harb2022faster}. Among them, \cite{harb2022faster,danisch2017large} claim that their methods can converge into the locally-dense decomposition. And \cite{harb2022faster} also claims that Greedy++ will converge into it. As an iterative method, \emph{\method} also does it. Firstly, let's introduce the quadratic program(QP) formula of locally-dense decomposition in \cite{danisch2017large,harb2022faster}.

\begin{equation}
    \label{eq:decomposition}
    \centering
    \begin{aligned}
        \textrm{minimize} \qquad  & \sum_{v \in \nodes}{\ell_v^2} \\
        \textrm{subject to} \qquad & f_e(u) + f_e(v) = 1 \qquad \forall e = uv \in \edges \\
                    & \ell_v = \sum_{e \ni v}{f_e(v)} \le D \qquad \forall v \in \nodes\\
                    & f_e(u) \ge 0 \qquad \forall e = uv \in \edges \\
                    & f_e(v) \ge 0 \qquad \forall e = uv \in \edges \\
    \end{aligned}
\end{equation}

The relationship between QP \eqref{eq:decomposition} and LDD is: When \emph{\method} makes $\sum_{v \in \nodes}{{\ell_v}^2}$ converge to the minimum, then the solution $(\bm{f},\bm{\ell})$ converges to the solution of locally-dense decomposition. 

In algorithm \ref{alg:lowd}, \emph{\method} will decide an update step d and redistribute the weight of edge e to decrease the optimization objective of QP \eqref{eq:decomposition} as much as possible. For example, for $e=(v_1,v_2)$ where $\ell_{v_1}>\ell_{v_2}$ and $f_e(v_1)>0$, we use $\bm{\ell}^{'}$ to represent the node loads after updating, i.e., ${\ell_{v_1}}^{'}\gets \ell_{v_1}-d$, ${\ell_{v_1}}^{'}\gets \ell_{v_1}+d$ and ${\ell_{v}}^{'} \gets \ell_v$ for $v\neq v_1,v_2$. Then:

\begin{center}
    \vspace{-0.15in}
    \begin{equation*}
        \begin{aligned}
            \sum_{v \in \nodes}{{{\ell_{v}}^{'}}^{2}}&=\sum_{v \in \nodes,v\ne v_1,v_2}{{{\ell_{v}}^{'}}^{2}}+{{\ell_{v_1}}^{'}}^{2}+{{\ell_{v_2}}^{'}}^{2}\\
            &=\sum_{v \in \nodes,v\ne v_1,v_2}{{\ell_{v}}^{2}}+{(\ell_{v_1}-d)}^2+{(\ell_{v_2}+d)}^2\\
            &=\sum_{v \in \nodes}{{\ell_{v}}^2}+2d\cdot(d+\ell_{v_2}-\ell_{v_1})<\sum_{v \in \nodes}{{\ell_{v}}^2}\\
        \end{aligned}
        \label{eq:decrease}
    \end{equation*}
\end{center}

In algorithm \ref{alg:lowd}, update step $d=\min \{(\ell_{v_1}-\ell_{v_2})/2, \, f_{e}(v_1) \}$ then $d+\ell_{v_2}-\ell_{v_1}<0$, and in this case $d>0$. The less-than sign in the last line holds true.

The local optimality of \emph{\method} is because it decreases the optimization objective of QP \eqref{eq:decomposition} and LP \eqref{eq:tighter_dual} as much as possible. In the above equation of QP \eqref{eq:decomposition}, $d=(\ell_{v_1}-\ell_{v_2})/2$ can minimize $2d\cdot(d+\ell_{v_2}-\ell_{v_1})$ according to the mean inequality so that $\sum_{v \in \nodes}{{{\ell_{v}}^{'}}^2}$ is the minimum, and there should be $d\ge f_e(v_1)$ to satisfy constraint $f_e(v_1)\ge 0$, then $d=\min \{(\ell_{v_1}-\ell_{v_2})/2, \, f_{e}(v_1) \}$ is a locally optimal operation which satisfies the constraint and decrease $\sum_{v \in \nodes}{{\ell_{v}}^2}$ as much as possible. As for LP \eqref{eq:tighter_dual}, \emph{\method} is also locally optimal to minimize $D$.

\begin{theorem}
    \emph{\method} will optimize QP \eqref{eq:decomposition} until $\sum_{v \in \nodes}{{l_v}^2}$ converge to the minimum.
    \label{th:decomposition}
\end{theorem}

\begin{proof}
    First, any edge redistribution in \emph{\method} will decrease the optimization objective of QP \eqref{eq:decomposition} as we claimed before. According to Theorem \ref{th:tmct}, $\sum_{v \in \nodes}{{\ell_{v}}^2}$ will converge to some value. Suppose it is not the minimum, if \emph{\method} stops changing any edge redistribution, the node loads in $\graph$ must be satisfied with property \ref{prop:one-way}, i.e., if two endpoints have different loads, the edge connecting them should be only distributed in one-way to the node with a lower load. Otherwise, \emph{\method} can continue its iterative operation to decrease $\sum_{v \in \nodes}{{{\ell_{v}}^{'}}^{2}}$. Then, the nodes with the same loads will consist of new nodesets $B_i$, from a global perspective, it will result in a new sequence $\emptyset=B'_0 \subsetneqq B' _1 \subsetneqq B' _2 \subsetneqq ... \subsetneqq B' _k=\nodes$ with $\lambda' _1>\lambda' _2>...>\lambda' _k$. According to property \ref{prop:one-way}, $B' _i=\mathop{\arg\max}\limits_{W \supsetneqq B' _i-1}{\frac{\weights(\edges(W))-\weights(\edges(B' _{i-1}))}{|W\setminus B' _{i-1}|}}$ so it is satisfied with the definition of locally-dense decomposition. Now that we suppose this decomposition doesn't converges to the minimum of QP \eqref{eq:decomposition}, it must be another LDD with a different {$\bm{\ell}^*$}, which contradicts with the property \ref{prop:uniquity1} and \ref{prop:uniquity2}.
    
    If \emph{\method} doesn't stop, it must decrease the optimization objective of QP \eqref{eq:decomposition} infinitesimally. If so, edge weights must be changed infinitesimally when $T \to \infty $, otherwise the optimization objective will not decrease infinitesimally, which makes a contradiction to convergence. Given that \emph{\method} manages to balance loads of two endpoints connected by an edge, node loads in the locally-dense decomposition must obey property \ref{prop:one-way} in the sense of limit because edge weights only can be changed infinitesimally. At last, these node loads will produce a new sequence $\emptyset=B_0 \subsetneqq B _1 \subsetneqq B_2 \subsetneqq ... \subsetneqq B_k=\nodes$ in the limit sense, which makes a contradiction similarly as above.

    Therefore, the optimization objective in QP \eqref{eq:decomposition} will converge to the minimum, and the minimum represents the graph converges into the locally-dense decomposition.
\end{proof}

Therefore, \emph{\method} can make $(\bm{f},\bm{\ell)}$ converge into locally-dense decomposition. LDD is an important basic problem for many variants of DSP. \cite{ma2022finding} use it to detect a variant of DSP called locally densest subgraph. Besides, we provide a perspective on the relationship between LDD and another variant of DSP concerning the densest subgraph with size constraint, called \textit{densest k-subgraph} (DkS) and \textit{at-least-k subgraph} (DalkS) problems. Its proof is provided in the appendix.  

\begin{corollary}
     \label{coroll:dks}
    The following results hold up on DkS and DalkS problems:
    \begin{enumerate}[label={\arabic*.}]
        \item For $k=|B_j|\,,\forall j \in \left \{ 1,2,...,k \right \} $, the DkS (or DalkS) is just the subgraph composed of nodes in $B_j$.
        \item For $|B_{j-1}|<k<|B_j|\,, \forall j \in \left \{ 1,2,...,k \right \}$, the upper bound of density in DkS (or DalkS) is $\frac{\sum_{i=0}^{j-1}{\lambda_{i}*|B_i|}+(k-|B_{j-1}|)*\lambda_{j}}{k}$.
    \end{enumerate}
\end{corollary}

\subsection{Pruning pre-process}

In our experiments, \emph{\method} can detect the densest subgraph much faster than other baselines without pruning, but it is more efficient to use a pruning technique to locate the densest subgraph before using any iterative algorithm. And the pruning is also a subprocess of Greedy, which is also useful to speed up other approximation and exact algorithms.

First, we have the following necessity condition about the optimal solution for DSP, i.e., the optimal set $\optset$.

\begin{theorem}[Lower Bound\cite{khuller2009finding}]
    For each node $v \in \optset$ of the densest subgraph, $\setndeg{\optset}{v} \ge \rho(\optset)$.
    \label{th:lowerbound}
\end{theorem}

Therefore, we can conclude that $\setndeg{\mathcal{H}}{v} \ge \setndeg{\optset}{v} \ge \rho(\optset) \ge \rho(\mathcal{H})$ if $\optset \subset \mathcal{H}$, which means $v\not\in \optset$ if $v\in \mathcal{H}$ and $\setndeg{\mathcal{H}}{v}<\rho(\mathcal{H})$. Then we can use the density of a subgraph $\mathcal{H}$ as the lower bound to filter out the candidates of $\optset$ as long as the density of $\mathcal{H}$ can be easily obtained.  The pruning technique first estimates the lower bound based on the density of the remaining subgraph $\mathcal{H}$ at the current time and deletes nodes whose degrees are lower than $\rho(\mathcal{H})$ iteratively and their adjacent edges, then it updates the bound using updated $\rho(\mathcal{H})$. Details are in Algorithm \ref{alg:pruning}.

\begin{algorithm}[t]
    \SetKwFunction{DSPSolver}{DSPSolver}
    \caption{\textsc{Pruning}}
    \label{alg:pruning}
    \KwIn{Undirected graph $\graph$; plug-in \DSPSolver: \{ \textsc{Maxflow}, \textsc{greedy}, \textsc{greedy++}, \method, \textsc{Frank-Wolfe}, \textsc{FISTA}, etc. \}.}
    \KwOut{Solution of the densest subgraph of $\graph$.}

    $\mathcal{H} \leftarrow \graph$;

    $\delta \leftarrow \rho(\mathcal{H})$;

    
    \While{$\exists \, \mathcal{H^{'}} \subseteq  \mathcal{H}$ with $\setndeg{\mathcal{H}}{u} < \delta \quad \forall u \in \mathcal{H^{'}}$}{
            
            Remove all nodes in $\mathcal{H^{'}}$ and all its associated edges from $\mathcal{H}$;
            
            $\delta \leftarrow \rho(\mathcal{H})$;
    }
    
    $\subnode \leftarrow$ \DSPSolver{$\mathcal{H}$};    
    \hfill \linecomment{plug-in DSP solver}
    
    \Return{$\graph(\subnode)$.}
\end{algorithm}

This pruning technique is similar to pruning 1 in \cite{fang2019efficient}. However, we don't consider k-core explicitly and we modify the data structure Counting Sort to speed it up on unweighted graphs, which assigns it a specific time complexity, i.e., $O(M+N)$. On an unweighted graph, first, we calculate the degree of all nodes and sort them by Counting Sort, i.e., we record nodesets of each degree 0-$d_{max}$, and update the degrees of remaining nodes and corresponding sets after each iteration. In each round of scanning, we only need to scan all the sets in the range of $\left [ \left  \lfloor LowBound_{1}  \right \rfloor, \left \lceil LowBound_{2}-1  \right \rceil  \right ] $ , which are exactly the nodes to be deleted in the next iteration. $LowBound_{1}$ is the bound before updating $\rho(\mathcal{H})$ while $LowBound_{2}$ is the bound after updating $\rho(\mathcal{H})$. Although nodes with degrees less than $LowBound_{1}$ will appear after deleting nodes in each iteration, we can just put them into the set $\left \lfloor LowBound_{1} \right \rfloor $. Throughout the process, we check all the nodesets in the range of $\left [ 0,d_{max}  \right ] $ and only repeat searching $\lfloor LowBound_{1} \rfloor$ at most $T$ times. It is obvious to know $d_{max}$ and $T$ are both lower than $N$. So the time complexity of checking the nodeset of each degree is $O(N)$ on an unweighted graph. Therefore, the time complexity of the pruning on an unweighted graph is $O(M+N)$. 

On a weighted graph, after each iteration, we have to traverse the remaining nodes to find the nodes that can be deleted in the next iteration. Then the time complexity of the pruning on weighted graphs is $O(M+TN)$, where $T$ is the number of iterations until it stops.

Their difference in time complexity is because we use modified Counting Sort to make sure each node is only checked one time on unweighted graphs. 

\begin{complexity}
    For the pruning technique, its time complexity is $O(M+N)$ on unweighted graphs and $O(M+TN)$ on positive weighted graphs.
\end{complexity}

It seems that on weighted graphs it is not efficient because we don't know iteration count $T$ in advance, in fact in our experiments the time consumption of these two versions on unweighted and weighted graphs doesn't differ a lot. It is also an efficient solution that we can set an upper bound for $T$ on weighted graphs.

Next, we claim that pruning is a specific subprocess of Greedy as below. Lemma \ref{lem:k-core} is a crucial property of Greedy and we use it to prove Theorem \ref{th:pruning}. The definition of k-core and their proofs are in the appendix.

\begin{lemma}
     For any k, k-core can be achieved by the Greedy algorithm in \cite{charikar2000greedy}.
    \label{lem:k-core}
\end{lemma}
\begin{theorem}
    In Greedy\cite{charikar2000greedy}, when the density decreases for the first time, the remaining subgraph is the subgraph $\mathcal{H}$ when the pruning in Algorithm \ref{alg:pruning} stops iterations.
    \label{th:pruning}
\end{theorem}

In the above theorem, we can see that although the deletion rule is different in the pruning and Greedy. Pruning will get the same subgraph as the end of the monotonic increase of density in Greedy. Given that the pruning doesn't need to know which node has the lowest degree, it can achieve a faster deletion speed than Greedy.

In fact, the pruning is so efficient that it can delete most nodes and ensure that the densest graph is in the remaining subgraph with a much smaller size. 

%% file: 050exps.tex
\begin{table}[t]
    \centering
    \caption{\textnormal{Statistical information of Datasets.}}
    \label{tab:datasets}
    \vspace{-0.1in}
    \begin{tabular}{l|c|r} \toprule
      \textbf{Dataset} & \textbf{$ |\nodes|$} & \textbf{$|\edges|$} \\ \midrule
        ca-HepPh                & 12006                 & 118,489       \\
        comm-EmailEnron         & 36692                 & 183,831       \\
        ca-AstroPh              & 18771                 & 198,050       \\
        PP-Pathways             & 21267                 & 338,636       \\
        soc-sign\_slashdot      & 77350                 & 468,554       \\
        soc-sign\_epinion       & 131828                & 711,210     \\
        soc-Twitter\_ICWSM      & 465017                & 833,540       \\
        rating-StackOverflow    & {[}545195,96678{]}    & 1,301,942     \\
        ego-twitter             & 81306                 & 1,342,296     \\
        soc-Youtube             & 1134890               & 2,987,624     \\
        comm-WikiTalk           & 2394385               & 4,659,565     \\
        nov\_user\_msg\_time    & {[}2748001,8083629{]} & 48,078,692    \\
        cit-Patents\_AMINER     & 6840994               & 54,022,588    \\
        soc-Twitter\_ASU        & 11316811              & 63,555,749    \\
        soc-Livejournal         & {[}3201203,7489073{]} & 112,307,385   \\
        soc-Orkut               & 3072441               & 117,185,083   \\
        soc-SinaWeibo           & 58655849              & 261,321,033   \\
        \hline
        wang-tripadvisor        & {[}145316,1759{]}     & 175,755       \\
        rec-YelpUserBusiness    & {[}45982,11538{]}     & 229,906       \\
        bookcrossing            & {[}77802,185955{]}    & 433,652       \\
        librec-ciaodvd-review   & {[}21019,71633{]}     & 1,625,480     \\
        movielens-10m           & {[}69878,10677{]}     & 10,000,054    \\
        epinions                & {[}120492,755760{]}   & 13,668,320    \\
        libimseti               & 220970                & 17,233,144    \\
        rec-movielens           & {[}283228,193886{]}   & 27,753,444    \\
        yahoo-song              & {[}1000990,624961{]}  & 256,804,235   \\ \bottomrule
    \end{tabular}
\end{table}

\begin{table*}[htbp]
    \centering
    \caption{\textnormal{Running time (sec.) comparison for utilizing the pruning pre-process.}}
    \label{tab:time}
    \vspace{-0.1in}
    \resizebox{0.98\textwidth}{!}{
    \begin{tabular}{l|c|c||c|c|c||c|c||c|c} \toprule
        \textbf{Dataset}& w$\_$Pruning & uw$\_$Pruning  & BBST & Priority Tree & Pruning+PT & Max-Flow & w$\_$Pruning+M & Doubly-Linked List & uw$\_$Pruning+DLL \\ \midrule
        ca-HepPh        & 0.008           & 0.008   & 0.196    & 0.033   & \textbf{0.016}   & 1.351   & \textbf{0.145}    & 0.013   & \textbf{0.009}   \\
        comm-EmailEnron          & 0.016  & 0.017   & 0.377    & 0.066   & \textbf{0.025}   & 2.834   & \textbf{0.272}    & 0.03    & \textbf{0.019}   \\
        ca-AstroPh      & 0.013          & 0.013   & 0.354    & 0.065   & \textbf{0.043}   & 2.818   & \textbf{0.812}     & 0.025   & \textbf{0.017}   \\
        PP-Pathways  & 0.037 & 0.028  & 0.73     & 0.141   & \textbf{0.045}   & 4.92    & \textbf{0.412}      & 0.044   & \textbf{0.031}   \\
        soc-sign\_slashdot       & 0.041 & 0.041  & 1.112    & 0.187   & \textbf{0.072}   & 9.894   & \textbf{0.626}      & 0.082   & \textbf{0.046}   \\
        soc-sign\_epinion        & 0.063   & 0.064  & 1.877    & 0.31    & \textbf{0.111}   & 16.354  & \textbf{0.92}     & 0.15    & \textbf{0.073}   \\
        soc-Twitter\_ICWSM       & 0.168 & 0.15  & 4.807    & 0.537   & \textbf{0.202}   & 29.824  & \textbf{0.746}       & 0.288   & \textbf{0.156}   \\
        rating-StackOverflow     & 0.327 & 0.288  & 7.087    & 0.986   & \textbf{0.35}    & 55.803  & \textbf{0.515}      & 0.572   & \textbf{0.297}   \\
        ego-twitter     & 0.084         & 0.084   & 2.724    & 0.425   & \textbf{0.149}   & 24.217  & \textbf{1.357}      & 0.162   & \textbf{0.097}   \\
        soc-Youtube    & 0.569  & 0.577  & 18.292   & 2.035   & \textbf{0.631}   & 118.627 & \textbf{1.805}     & 1.118   & \textbf{0.585}   \\
        comm-WikiTalk   & 0.94   & 0.904   & 60.556   & 3.399   & \textbf{1.053}   & 230.949 & \textbf{2.862}     & 2.08    & \textbf{0.939}   \\
        nov\_user\_msg\_time     & 8.68  & 9.987   & 1317.44  & 58.34   & \textbf{9.268}   & $\ge$6h     & \textbf{35.823}     & 26.425  & \textbf{10.08}   \\
        cit-Patents\_AMINER   & 3.796  & 3.865    & 536.038  & 30.414  & \textbf{6.748}   & $\ge$6h     & \textbf{98.012}   & 18.005  & \textbf{4.967}   \\
        soc-Twitter\_ASU  & 6.68   & 6.867     & 1169.1   & 42.905  & \textbf{8.026}   & $\ge$6h     & \textbf{48.612}  & 20.659  & \textbf{7.22}    \\
        soc-Livejournal  & 9.676  & 10.326  & 1239.405 & 67.74   & \textbf{16.951}  & $\ge$6h     & \textbf{405.168} & 32.478  & \textbf{11.861}  \\
        soc-Orkut        & 9.693  & 10.173  & 506.536  & 80.45   & \textbf{13.225}  & 5274.2  & \textbf{339.112}   & 34.934  & \textbf{10.748}  \\
        soc-SinaWeibo    & 51.417 & 46.132  & 32197.2  & 225.065 & \textbf{56.184}  & $\ge$6h     & \textbf{252.011}   & 146.471 & \textbf{48.488}  \\
        \hline
        wang-tripadvisor       & 0.052 &     --   & 0.772    & 0.122   & \textbf{0.059}   & 6.203   & \textbf{0.278}     &  --       &  --                \\
        rec-YelpUserBusiness     & 0.026 &   --    & 0.61     & 0.105   & \textbf{0.04}    & 5.305   & \textbf{0.433}        &   --      & --                 \\
        bookcrossing       & 0.107  &    --   & 1.979    & 0.341   & \textbf{0.122}   & 17.467  & \textbf{0.426}   &    --         &    --              \\
        librec-ciaodvd-review    & 0.093  &  --   & 3.245    & 0.543   & \textbf{0.152}   & 41.812  & \textbf{1.629}        &    --     &   --               \\
        movielens-10m       & 0.3  &   --     & 18.088   & 2.781   & \textbf{1.433}   & 412.519 & \textbf{65.712}       & --         &     --             \\
        epinions     & 0.943  &   --  & 36.722   & 5.685   & \textbf{1.383}   & 624.063 & \textbf{17.495}        & --        &      --            \\
        libimseti       & 0.844  &   --   & 78.901   & 7.302   & \textbf{1.767}   & 942.125 & \textbf{67.012}        &  --       &   --               \\
        rec-movielens     & 0.98  &  --    & 27.152   & 9.493   & \textbf{3.587}   & 1376.71 & \textbf{135.278}      &   --      &      --            \\
        yahoo-song     & 10.936  & --   & 691.58   & 125.762 & \textbf{27.115}  & $\ge$6h     & \textbf{2732.31}    & --        &      --         \\ 
        \hline
        \multicolumn{10}{l}{\multirow{2}{0.98\textwidth}{\textbf{note:}`w\_' and `uw\_' denote the weighted and unweighted version; `Pruing+PT', `Pruning+M', and `Pruing+DLL' denote using the pruning to process the dataset at first and following with Priority Tree, Max-Flow, and Doubly-linked List, respectively.}} \\
        \multicolumn{10}{l}{} \\
        \bottomrule
    \end{tabular}
    }
\end{table*}

We design experiments to answer the following questions:
\begin{compactitem}
    \item{\textbf{Effectiveness Of The Pruning Technique:}} How significantly can the pruning reduce the size of the graph? How about its speedup to the approximation and exact algorithms?
    \item{\textbf{Effectiveness of \emph{\method} When Solving DSP:}} Does \emph{\method} detect the densest subgraph faster than other iterative algorithms?
    \item{\textbf{Effectiveness of \emph{\method} When Solving LDD:}} Does \emph{\method} optimize the linear programming of locally-dense decomposition faster than other baselines?
\end{compactitem}

\textbf{Datasets\&Implementation.}
We collected 26 networks from popular publicly available repositories, including Stanford’s SNAP database~\cite{jure2014snapnets}, AMiner scholar datasets~\cite{wan2019aminer}, Network Repository~\cite{nr2015aaai}, ASU’s Social Computing Data Repository~\cite{ZafaraniLiu2009}, and Konect~\cite{kunegis2013konect} etc. Multiple edges, self-loops are removed, and the directionality is ignored for directed graphs. Table~\ref{tab:datasets} lists their statistical information, where the first group (17 in total) are unweighted and the remaining are weighted. All the experiments are performed on a machine with 2.4GHz Intel(R) Xeon(R) CPU(8 cores) 
and 500GB of RAM. All baselines are implemented in C++ 14. The sizes of the networks range from 10 thousand to 10 million.

\begin{figure}[htbp]
    \centering
    \includegraphics[width=0.7\linewidth]{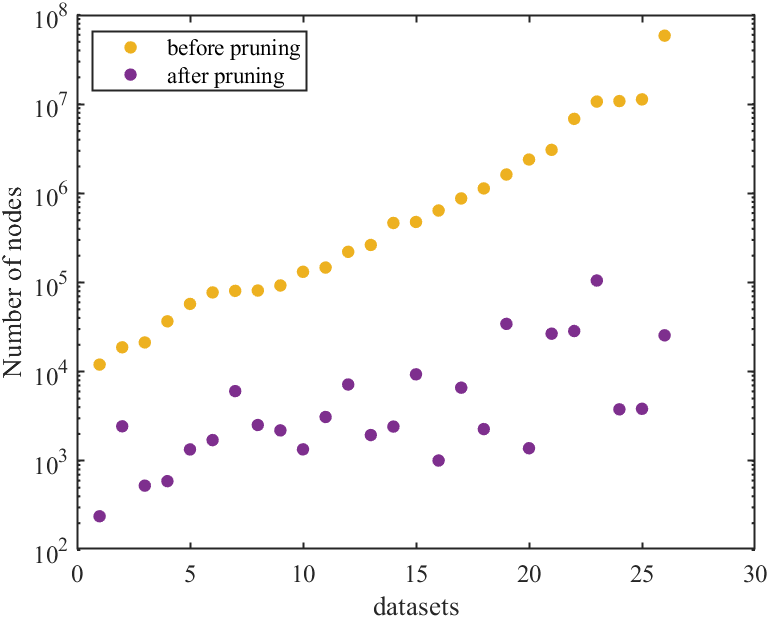}
    \caption{The pruning pre-process reduces the graph size significantly.}
    \label{fig:pruning}
\end{figure}

\subsection{ Effectiveness of The Pruning Technique}
We choose Max-Flow\cite{goldberg1984finding}, Greedy in priority tree(PT)\cite{hooi2016fraudar}, Greedy in doubly-linked list(DLL)\cite{boob2020flowless} and Greedy in Balanced Binary Search Tree(BBST)\cite{boob2020flowless} as baseline algorithms. PT, DL, BBST are different data structures used in Greedy, DLL can only be applied on unweighted graphs while BBST and PT can be applied on weighted graphs. We implement Max-Flow, Greedy in PT by ourselves, the code of Greedy in DLL and BBST is from \cite{boob2020flowless} because they optimize it enough well. 

Although Core-App\cite{fang2019efficient} can attain $k_{max}$-core efficiently which is also 1/2-approximation like Greedy, its density is no more than the density attained by Greedy because Greedy can choose the subgraph with the maximum density and $k_{max}$-core is attained at Greedy's some time according to Lemma \ref{lem:k-core}. Therefore, among approximation algorithms, we only select Greedy in different data structures because they can achieve
equivalent results.

Figure~\ref{fig:pruning} compares the size of graphs before and after the pruning pre-process. We can see that applying the pruning leads to about two orders of magnitude of reduction for the size of all networks. The effect of pruning is more obvious on larger graphs.

\begin{table*}[t]
    \centering
    \caption{Comparison between \emph{\method} and baselines.}
    \label{tab:lowd}
    \vspace{-0.1in}
    \resizebox{0.95\textwidth}{!}{
        \begin{tabular}{l|l|l|l|l|l|l|l|l|l|l|l}\toprule
        \multicolumn{2}{c}{}&\multicolumn{5}{|c|}{Running time}&\multicolumn{5}{c}{The number of iteration rounds}\\
        \midrule
        \textbf{Dataset} &
        \textbf{$|\mathcal{H}^{\prime}|$} & \textbf{LOWD} & \textbf{Greedy++} & \textbf{FW} & \textbf{FISTA} & \textbf{FW-M} & \textbf{LOWD} & \textbf{Greedy++} & \textbf{FW} & \textbf{FISTA} & \textbf{FW-M}\\
        \midrule
        ca-HepPh                 & 239          & 0.0019      & \first{0.0013}     & 0.0018        & 0.0122    & \second{0.0014}    & 1             & 1            & 1               & 1         & 1       \\
        comm-EmailEnron          & 592          & \first{0.0115}      & 0.0187     & 0.0645        & 0.2437    & \second{0.0163}    & 13            & 10           & 78              & 56       & 19      \\
        ca-AstroPh               & 2441         & \first{0.1202}      & \second{0.1759}     & 0.7566        & 2.2657    & 0.2125    & 36            & 28           & 345             & 183       & 83      \\
        PP-Pathways                    & 527          & \second{0.0089}      & \first{0.0029}     & 0.0354        & 0.1023    & 0.0241    & 6             & 1            & 24              & 18        & 16      \\
        soc-sign\_slashdot       & 1709         & \second{0.0255}      & \first{0.0045}     & 0.0807        & 0.3132    & 0.0361    & 10            & 1            & 42              & 27        & 18      \\
        soc-sign\_epinion        & 1345         & \first{0.0178}      & 0.0830      & 0.1309        & 0.2721    & \second{0.0700}      & 4             & 12           & 43              & 15        & 22      \\
        soc-Twitter\_ICWSM            & 2423         & \first{0.0944}      & 0.5038     & 4.9809        & 1.3091     & 0.9862    & 39            & 120          & 2820            & 115       & 573     \\
        rating-StackOverflow   & 1008         & \first{0.0199}      & 0.0666     & 0.1635        & 0.4072    & \second{0.0286}    & 24            & 45           & 272             & 125       & 44      \\
        ego-twitter                    & 2523         & \first{0.0929}      & 0.4272     & 1.9648        & 0.7245    & \second{0.0813}    & 13            & 46           & 478             & 25        & 19      \\
        soc-Youtube              & 2269         & \first{0.1077}      & 0.8683     & 1.7423        & 3.6148    & \second{0.3113}    & 29            & 127          & 608             & 181       & 106     \\
        comm-WikiTalk            & 1384         & 0.0099      & 0.0107     & \second{0.0094}        & 0.0412    & \first{0.0093}    & 1             & 1            & 1               & 1         & 1       \\
        nov\_user\_msg\_time          & 3779         & \second{3.026}       & 5.5366     & 12.2961       & 19.9690   & \first{2.1637}    & 96            & 72           & 436             & 91       & 76      \\
        cit-Patents\_AMINER            & 28546        & \first{2.3379}      & 29.3613    & 23.5359       & 44.4324   & \second{10.5380}    & 17            & 65           & 224             & 43        & 100     \\
        soc-Twitter\_ASU               & 3834         & 0.1189      & 0.2180      & \first{0.1034}        & 0.5913    & \second{0.1088}    & 1             & 1            & 1               & 1         & 1       \\
        soc-Livejournal & 105265       & \first{152.9460}     & 13795.667  & 24839.696     & 57567.91  & \second{1322.3124} & 437           & 8436         & 83549           & 17099     & 4615    \\
        soc-Orkut                & 26670        & \first{14.9269}     & 111.2410    & 1440.7987     & 1661.6384 & \second{86.1052}   & 71            & 150          & 8242            & 1063      & 495     \\
        soc-SinaWeibo          & 25556        & \first{36.0107}     & 308.7090    & 2694.3677     & 3020.8957 & \second{280.3085}  & 262           & 594          & 21266           & 2589      & 2263    \\
        \hline
        wang-tripadvisor                    & 3103         & \first{0.0708}      & 0.8382     & 1.7039        &    2.4322       & \second{0.6695}    & 121           & 187          & 4031            &    941        & 1403    \\
        rec-YelpUserBusiness  & 1343         & \first{0.0732}      & 1.5679     & 0.7741       &  3.7359          & \second{0.0987}    & 46            & 207          & 784             &   466        & 104     \\
        bookcrossing        & 1946         & \first{0.0703}      & 1.6941     & 1.0291        &    5.9866    & \second{0.1824}    & 75            & 278          & 1383            &  1516       & 250     \\
        librec-ciaodvd-review      & 2195         & \first{0.2090}       & 3.6041     & 14.9112       &    33.9615     & \second{0.2609}    & 49            & 139          & 4644            &   958     & 80      \\
        movielens-10m               & 6049         & 0.1322      & 0.5096     & \first{0.1245}        &  1.5077        & \second{0.1249}    & 1             & 1            & 1               &   1     & 1       \\
        epinions                     & 6616         & \first{1.2117}      & 140.5527   & 187.0548      &    268.6413    & \second{5.0036}    & 43            & 788          & 9081            &  1473     & 239     \\
        libimseti                           & 7179         & \first{1.1102}      & 247.5964   & 57.4351       &   384.2001    & \second{3.0025}    & 18            & 594          & 1142            &  912    & 59      \\
        rec-movielens             & 9335         & \second{8.8373}      & 839.4542   & 83.8077       &    635.9382     & \first{1.7434}    & 52            & 738          & 563             & 480       & 11      \\
        yahoo-song                          & 34352        & \second{155.6296}    & 11453.641  & 4144.8288     &    46170.844      & \first{51.6024}   & 154           & 1532         & 4959            &    4892    & 61      \\
        \hline
        \multicolumn{12}{l}{\multirow{1}{0.95\textwidth}{\textbf{note:} $|\mathcal{H}^{\prime}|$ denotes the size of the subgraph after the pruning, \textbf{FW} means Frank-Wolfe in \cite{danisch2017large}.}} \\
        \bottomrule
        \end{tabular}
    }
\end{table*}

We compare the running time of the approximation and exact algorithms on weighted and unweighted graphs in Table~\ref{tab:time}, where the results are the average of 5 trails. We report the density of detected subgraphs by approximation and exact algorithms in the appendix. The density obtained by the same type of algorithm is equivalent (We use $||$ to distinguish different types). It shows that the pruning can save a lot of time for subsequent algorithms, i.e., Greedy in different data structures and Max-Flow; the pruning achieves the best results (as the bold labeled) on both weighted graphs and unweighted graphs.

Columns 2 and 3 are the pruning used on weighted and unweighted graphs respectively (Of course an unweighted graph can be seen as a weighted graph with $w_e=1$ for any edge $e$). They prune the graph and lock the densest subgraph into a much smaller range.

Columns 4-6 show the results of approximation algorithms for all networks, we can find that PT is much faster than BBST. Although BBST has the same average time complexity as PT, i.e., $O(M \log N)$, it actually runs much slower than PT, especially on large datasets, which is attributed to the higher probability of a worse situation in BBST since deleting the node with the lowest weight iteratively in Greedy is not a balanced operation in BBST. Therefore, we just combine the pruning technique with PT, and it can detect the same subgraph faster,  which achieves about $3.563 \times$ on average by utilizing the pruning technique than using PT alone.

Although the time complexity of the pruning on weighted graphs is $O(M+TN)$ and we are not sure about the iteration count $T$, in our experiments, the number of deleted nodes approximately obeys the exponential law (in the appendix we plot 6 datasets to illustrate the exponential law). Therefore, $T$ is $O(\log N)$ according to experimental performance so the pruning is also efficient on weighted graphs, which can be confirmed with the results of columns 5 and 6, i.e., w\_Pruning+PT runs faster than using PT alone.

Columns 7-8 are the results of the exact algorithms. It shows that the pruning accelerates Max-Flow by $28.215 \times$ on average (ignoring the time-out datasets, although it performs better on large datasets), which can be attributed to two reasons: First, a higher density lower bound $\delta$ for the subsequent binary search. Second, a smaller graph needs to be explored next. These factors greatly reduce the number of binary-search iterations and the execution time of the max flow algorithm.

Columns 9-10 are the results of approximation algorithms for unweighted graphs. Greedy in doubly-linked list achieves $2.312 \times$ speedup on average by utilizing the pruning. Although the speed of uw\_Pruning is not faster than that of w\_Pruning evidently, it has an accurate time complexity, i.e., $O(M+N)$. By utilizing the pruning technique, we acelerate DLL about $2.312\times$ on average.

\subsection{ Effectiveness of \emph{\method} When Solving DSP}

\begin{figure*}[t]
    \centering
    \subfigure[soc-Twitter\_ICWSM]{\includegraphics[width=0.24\linewidth]{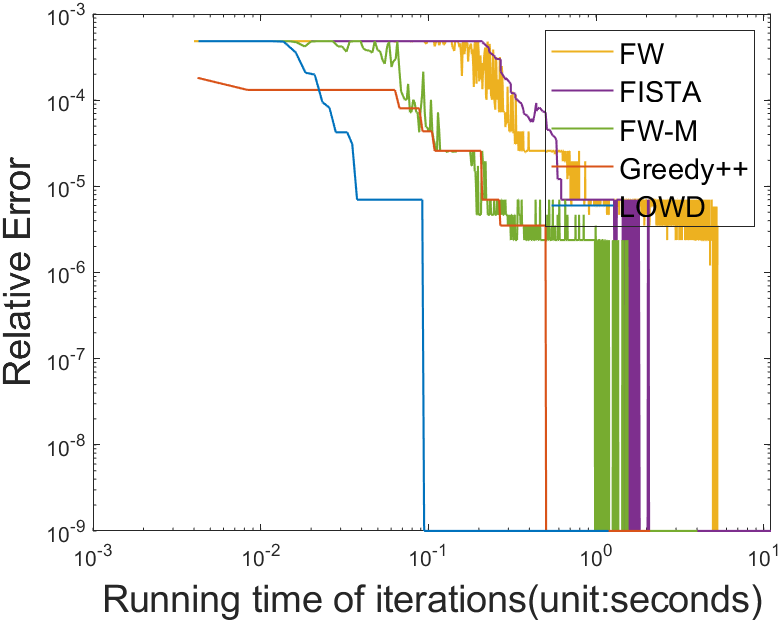}}
    \subfigure[ego-twitter]{\includegraphics[width=0.24\linewidth]{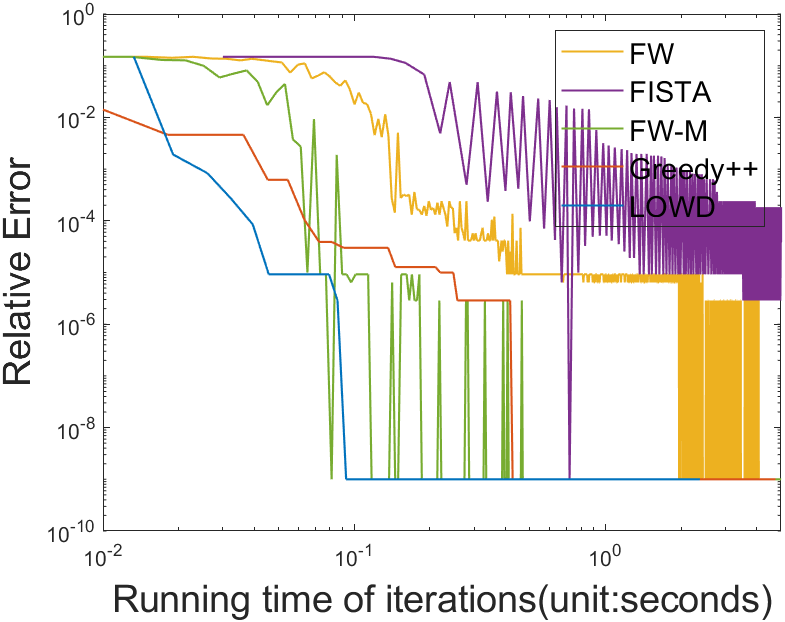}}
    \subfigure[soc-Youtube]{\includegraphics[width=0.24\linewidth]{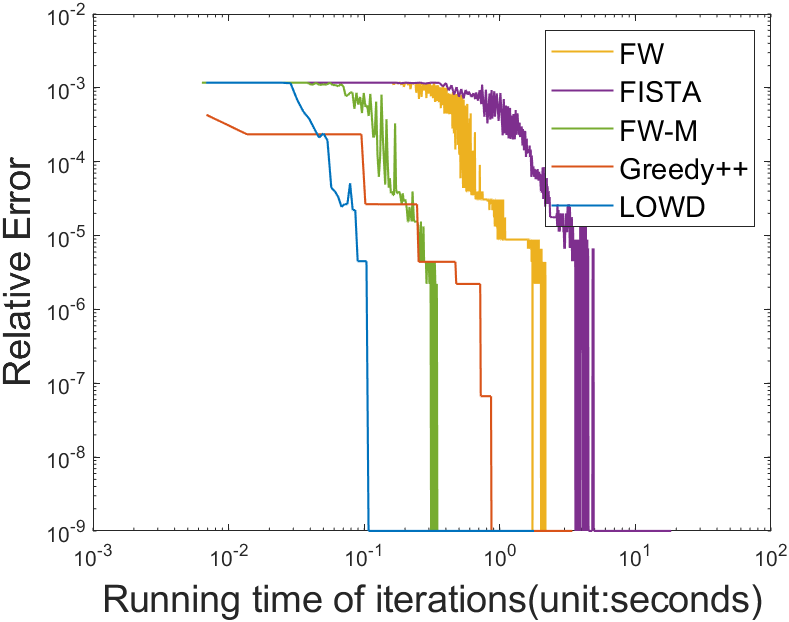}}
    \subfigure[soc-Livejournal]{\includegraphics[width=0.24\linewidth]{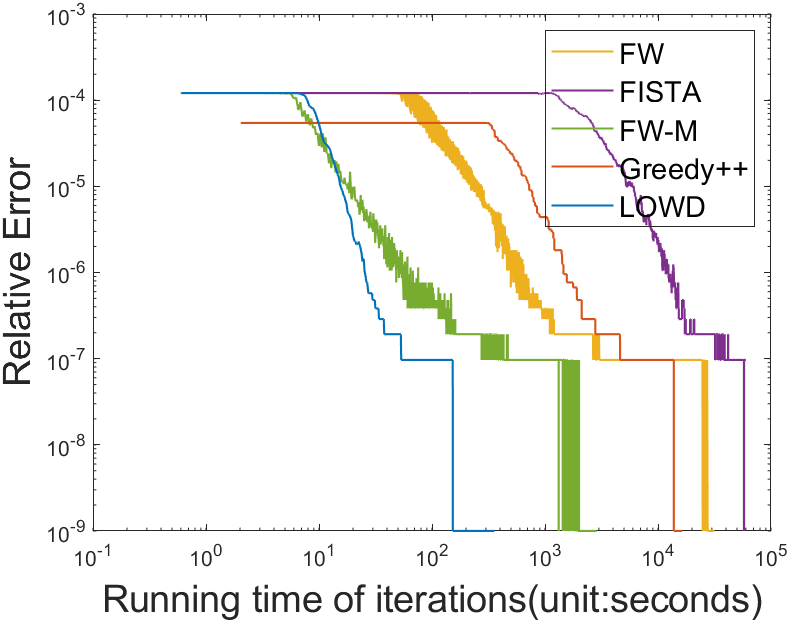}}
    \subfigure[soc-SinaWeibo]{\includegraphics[width=0.24\linewidth]{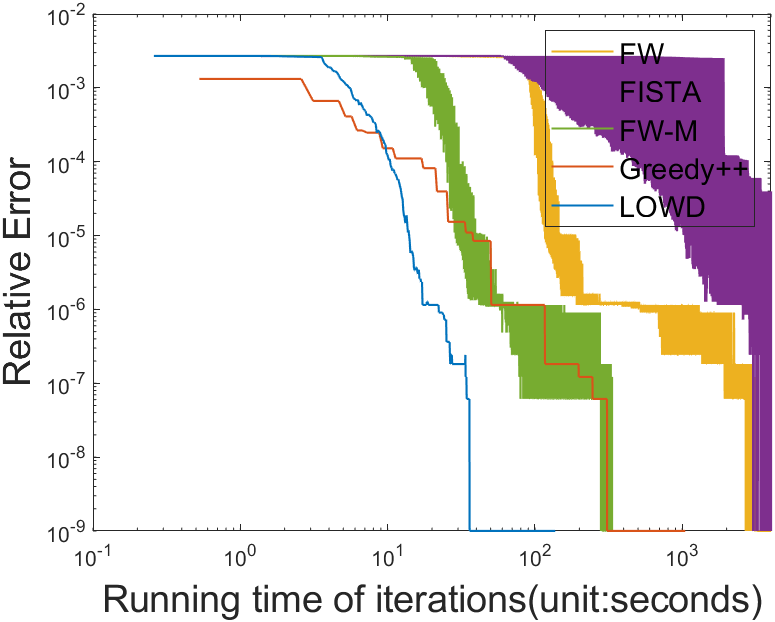}}
    \subfigure[wang-tripadvisor]{\includegraphics[width=0.24\linewidth]{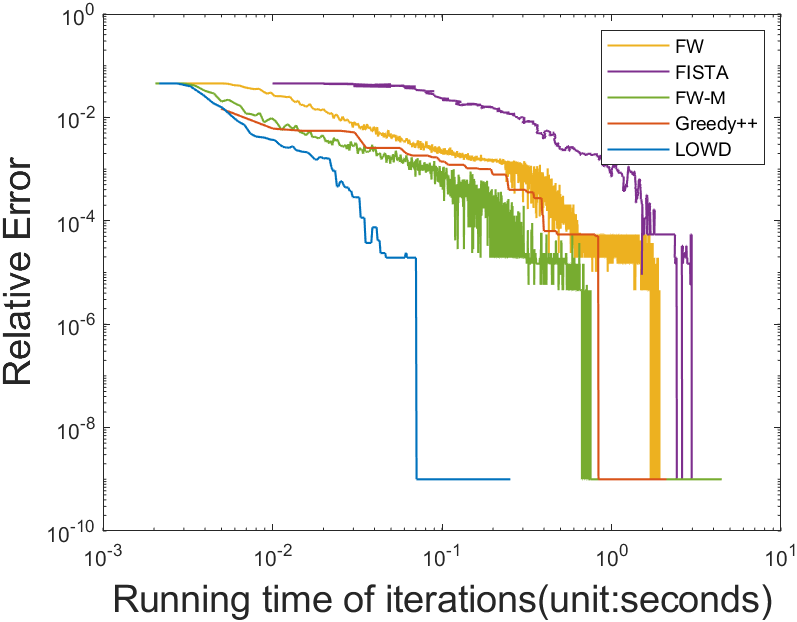}}
    \subfigure[bookcrossing]{\includegraphics[width=0.24\linewidth]{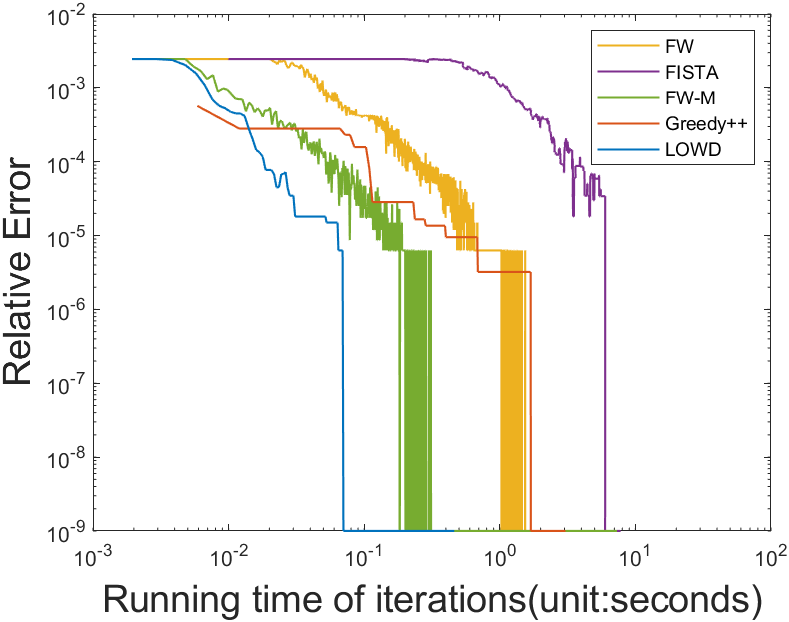}}
    \subfigure[yahoo-song]{\includegraphics[width=0.24\linewidth]{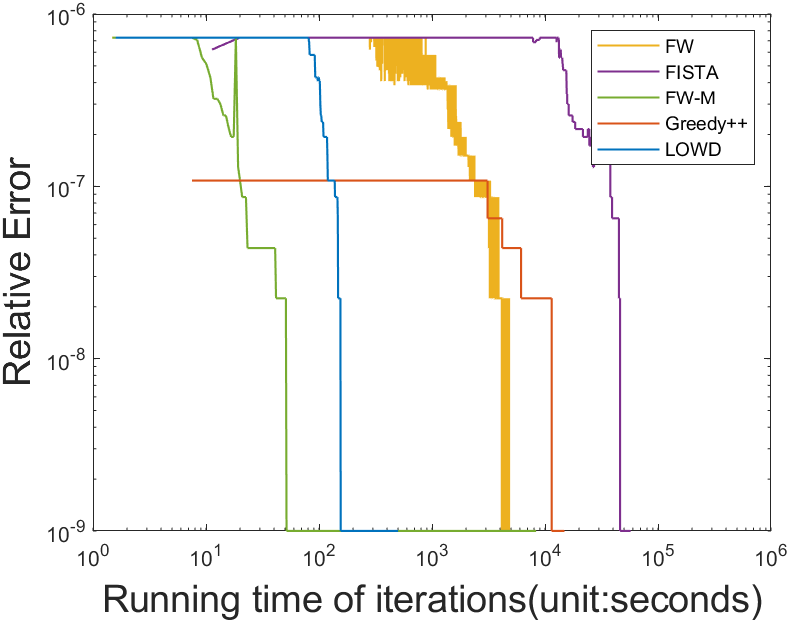}}
    \caption{Comparison on detecting the densest subgraph.}
    \label{fig:pruning_densest}
\end{figure*}

\begin{figure*}[htbp]
    \centering
    \subfigure[soc-Twitter\_ICWSM]{\includegraphics[width=0.24\linewidth]{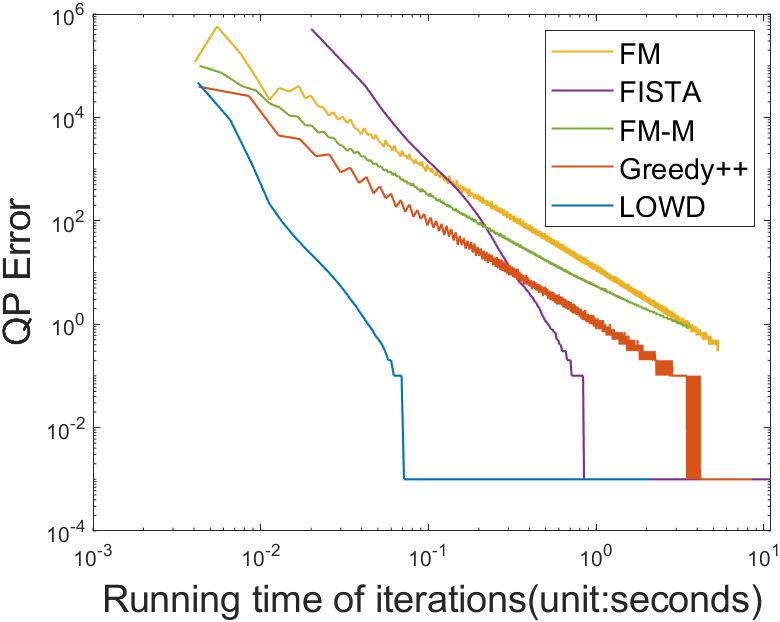}}
    \subfigure[ego-twitter]{\includegraphics[width=0.24\linewidth]{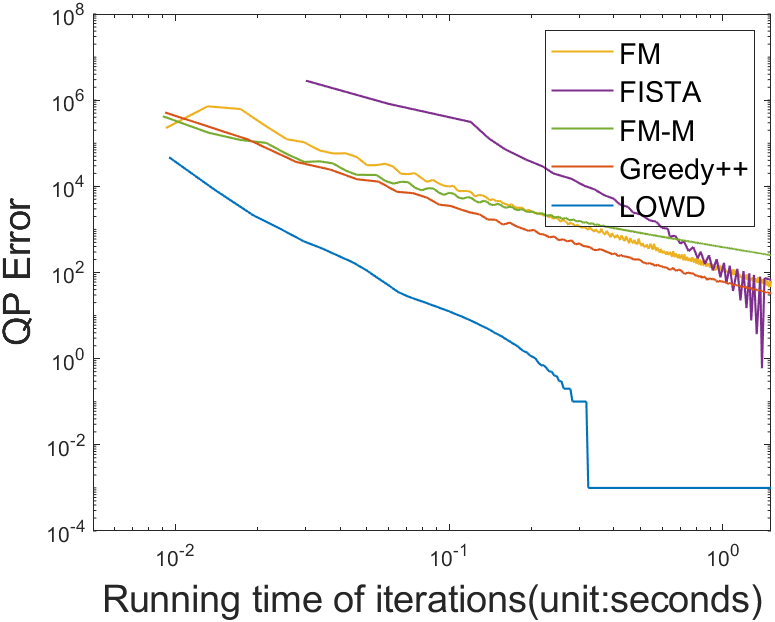}}
    \subfigure[soc-Youtube]{\includegraphics[width=0.24\linewidth]{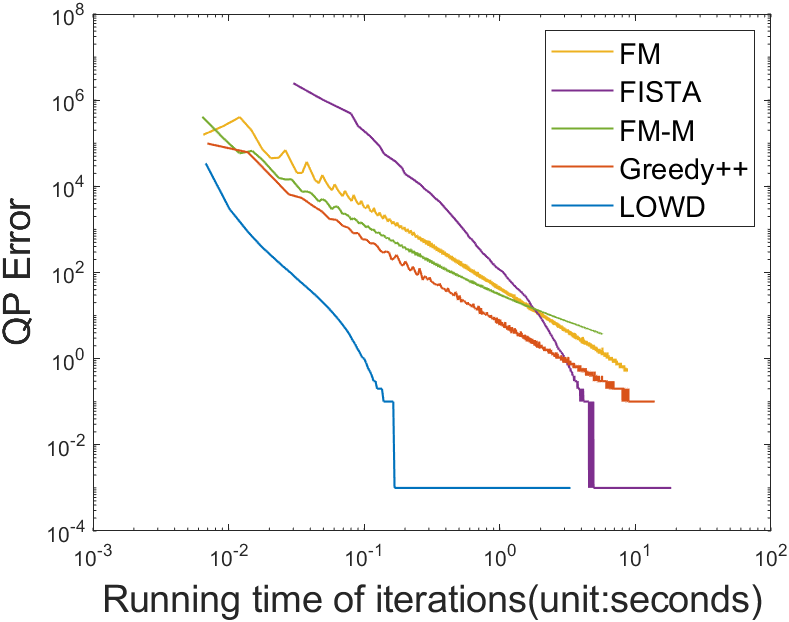}}
    \subfigure[soc-Livejournal]{\includegraphics[width=0.24\linewidth]{picture/pruning_LDD/soc-Youtube_SNAP.edgelist.png}}
    \subfigure[soc-SinaWeibo]{\includegraphics[width=0.24\linewidth]{picture/pruning_LDD/soc-Youtube_SNAP.edgelist.png}}
    \subfigure[wang-tripadvisor]{\includegraphics[width=0.24\linewidth]{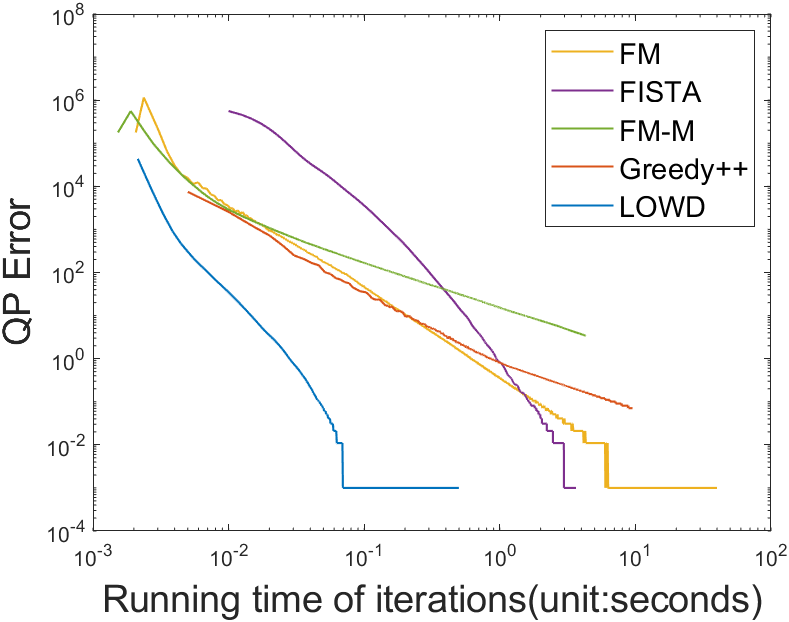}}
    \subfigure[bookcrossing]{\includegraphics[width=0.24\linewidth]{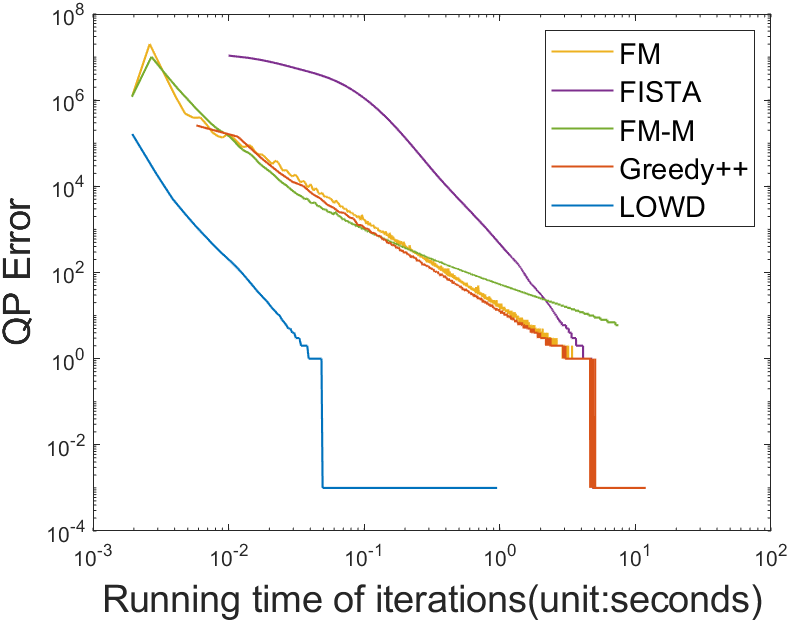}}
    \subfigure[yahoo-song]{\includegraphics[width=0.24\linewidth]{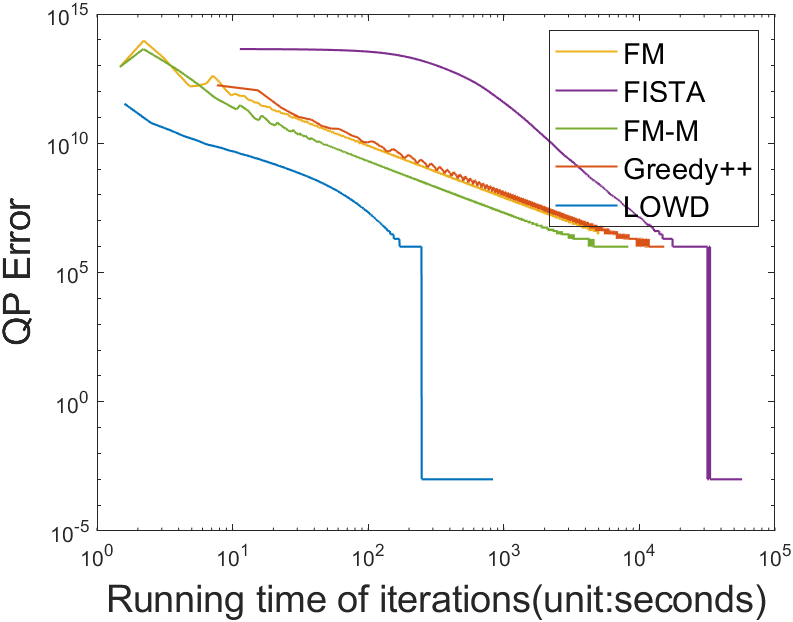}}
    \caption{Comparison on detecting the locally-dense decomposition.}
    \label{fig:pruning_ldd}
\end{figure*}

To fully demonstrate the efficiency of \emph{\method}, we compare it with the state-of-the-art iterative algorithms including Frank-Wolfe(FW) in \cite{danisch2017large}, Greedy++ in \cite{boob2020flowless} and FISTA in \cite{harb2022faster}, we also compare it with "Frank-Wolfe with a modified learning rate"(FW-M), the original learning-rate for Frank-Wolfe is $\frac{2}{(t+2)}$(t is the iteration number) while learning-rate in FW-M is $\frac{1}{(t+1)}$ inspired by \cite{harb2022faster}. FISTA and Greedy++ are implemented based on their own original codes, while Frank-Wolfe and its variant(FW-M) are implemented by ourselves. In this subsection, we want to figure out how well these iterative algorithms perform on detecting the densest subgraph and fully demonstrate the effectiveness of \emph{\method}. We devise our experiments as below:
\begin{compactitem}
    \item{\textbf{Comparison on detecting the densest subgraph.}} We execute it to figure out whether \emph{\method} can perform better than other iterative algorithms on detecting the densest subgraph.
    \item{\textbf{The relative error on the density as iterations count increases.}} We supplement it to exhibit the convergence speed of \emph{\method} and other iterative algorithms on 8 randomly selected datasets.
\end{compactitem}

All these experiments are executed with the pruning technique to pre-process these datasets. Corresponding experiments without the pruning are also executed whose results are listed in the appendix.

We initialize \emph{\method}, Frank-Wolfe, FISTA and FW-M by distributing the weight of each edge equally to its associated nodes like lines 1-4 in Algorithm \ref{alg:lowd}, which makes the node weight positively correlated with its degree, our intuitive consideration is that nodes with higher degrees should not be deleted first while it is unclear which node each edge should be distributed to. This sort method is also used in \cite{danisch2017large}.

In Table~\ref{tab:lowd}, we compared \emph{\method} and other baselines on the subgraphs after pruning in Algorithm \ref{alg:pruning} on unweighted and weighted graphs. The optimal result is bolded and underlined, while the suboptimal result is only underlined. It shows that \emph{\method} detects the densest subgraphs much faster than other iterative algorithms from the overall experimental performance, most of its results are optimal or suboptimal. Besides, Greedy++ is easier to detect the densest subgraph in only one iteration, which partly shows the effectiveness of Greedy. FW-M needs fewer iteration rounds than using its original learning-rate and FISTA needs more time to detect DSP in this sorting method(sorting node load and deleting the node with the lowest load one by one). However, it is worth mentioned that DSP is a subproblem of LDD. Actually FISTA has a better convergence speed than other baselines when solving LDD and FW-M doesn't converge faster than FW according to Figure \ref{fig:pruning_ldd}. And \cite{harb2022faster} devised a method called "fractional peeling" to detect the densest subgraph faster than the sorting method in our experiments. Related theory and experiments can be seen in \cite{harb2022faster}. To be accurate, FW-M suits the sorting method better while FW, FISTA don't suit it very much (which can be improved by "fractional peeling"). And our method \emph{\method} performs best both in DSP and LDD. The following paragraph will provide more detailed information.

Next, we randomly choose 8 datasets to exhibit the relative error of these algorithms as iterations count increases. The definition of relative error is:
$(\rho^*-\rho(\subnode))/\rho^*$. We use logarithmic coordinates to exhibit it more explicitly. We add a small amount \textit{eps=1e-9} to the relative error because when we detect the graph with density $\rho(\subnode)=\rho^*$, its relative error is 0 and doesn't appear on logarithmic coordinates. \emph{\method} has the optimal convergence speed on 7 datasets and the suboptimal convergence speed on 1 dataset, which is slower than FW-M. \emph{\method}'s relative error seldom vibrates because it monotonically decreases the optimization of LP \ref{eq:tighter_dual} and QP \ref{eq:decomposition} based on its local optimality. Greedy doesn't vibrate because it stores the best result of all iterations while the relative errors of Frank-Wolfe, FISTA and FW-M vibrate apparently because they are based on classical methods towards convex optimization like gradient descent and its variants. They don't ensure monotonic decline, which we will observe again in experiments about LDD.

\subsection{ Effectiveness of \emph{\method} When Solving LDD}

\begin{figure}[htbp]
    \centering
    \includegraphics[width=0.95\linewidth]{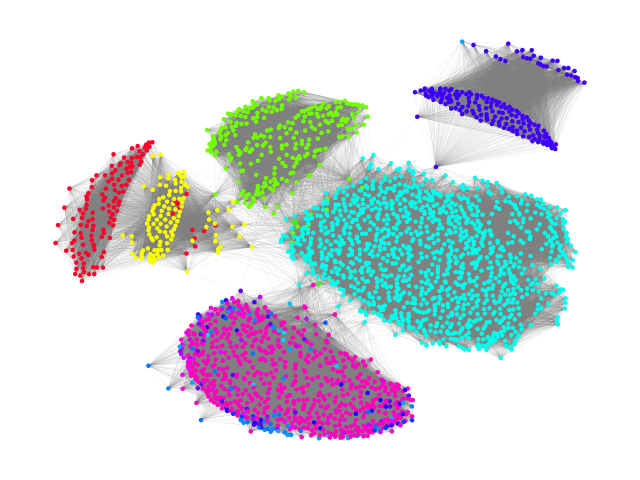}
    \vspace{-0.1in}
    \caption{LDD in ego-twitter (500 iterations of \emph{\method})}
    \label{fig:BDSS}
\end{figure}

We also conduct experiments about the convergence speed to the locally-dense decomposition in Figure \ref{fig:pruning_ldd}. The definition of QP error on the y-axis is: $\sum_{v \in \nodes}{{\ell_{v}}^{2}}-\sum_{v \in \nodes}{{\ell_{v}^{*}}^{2}}$, where $\sum_{v \in \nodes}{{\ell_{v}}^{2}}$ is attained by these iterative algorithms and $\sum_{v \in \nodes}{{\ell_{v}^{*}}^{2}}$ is the optimal value of QP \eqref{eq:decomposition}. We use the minimum $\sum_{v \in \nodes}{{\ell_{v}}^{2}}$ attained by these algorithms to represent the optimal value and add a small amount \textit{eps=1e-3} to QP error similarly.

From Figure \ref{fig:pruning_ldd}, we can see that \emph{\method} optimize $\sum_{v \in \nodes}{{\ell_{v}}^{2}}$ much faster than any other iterative algorithms because the x-axis is logarithmic coordinate. We conclude this to the local optimality of \emph{\method} in QP \eqref{eq:decomposition}. FISTA performs better on LDD than FW-M, Frank-Wolfe and Greedy++, which is consistent with the conclusion in \cite{harb2022faster}. 
In Figure \ref{fig:pruning_ldd} all baselines will vibrate to some degree, while our method \emph{\method} can optimize QP \eqref{eq:decomposition} monotonically.

To observe the locally-dense decomposition, we run \emph{\method} 500 iterations on \textit{ego-twitter} after pruning to achieve Figure \ref{fig:BDSS} as an example. We use colormap $gist\_rainbow$ in $matplotlib$ to dye nodes according to their loads, therefore, nodes with the same loads should have the same color. In Figure \ref{fig:BDSS} red means the minimum node load and violet means the maximum node load. And the nodes with maximum loads represent $B_1$ and the nodes with the second largest loads represent $B_{1}\setminus B_{2}$. We can really observe the hierarchical phenomenon on node loads in locally-dense decomposition. Besides, we checked some edges with $e=(u,v)$ and $u$,$v$ have different node loads, and these edge weights are all distributed in one-way to endpoints with lower loads. Another interesting discovery is that the node position is based on the Kamada-Kawai path-length in $matplotlib$, we can see nodes with similar loads are close to each other, which also implies that Kamada-Kawai path-length has connections with locally-dense decomposition.

By observing the locally-dense decomposition, it not only helps to detect different variants of the densest subgraph like DkS in Corollary \ref{coroll:dks} and locally densest subgraph in \cite{ma2022finding}, but also helps us to understand the whole graph about the theme of "edge-density". Given that we get the densest subgraph by distributing edge weights, we can understand the reason why some nodes with high degrees are not in the densest subgraph is that they have to distribute some edge weights to their neighbors with low loads. The influence is transmitted along edges, so the loads of two nodes can be influenced by each other as long as they are in the same connected branch. The previous work \cite{veldt2021generalized} has found this phenomenon, they generalized the Greedy algorithm of \cite{charikar2000greedy} with a parameter $p$ (the Greedy algorithm is equal to the generalized form with $p=1$) and set $p=1.05$ to spread the influence of node neighbors slightly, which helps to detect denser subgraphs than Greedy.

%% file: 020related.tex
The most recent surveys~\cite{lanciano2023survey,luo2023survey} present a systematic, thorough overview and
summarization of the densest subgraph problem, and the tutorials~\cite{Gionis2015DSD, fang2022densest} also
give a comprehensive survey of the discovery of the densest subgraphs on large graphs and 
discuss the challenges and various applications.

For the graph with non-negative edge weights, the densest subgraph can be identified in polynomial time by 
solving a maximum flow problem~\cite{goldberg1984finding, gallo1989fast, khuller2009finding}; 
Charikar~\cite{charikar2000greedy} introduces a linear programming formulation of the problem and shows that 
the greedy algorithm proposed by Asashiro et al.~\cite{asahiro2000greedily} produces 
a $\frac{1}{2}$-approximation of the optimal density in linear time. 
\cite{danisch2017large} devises an efficient algorithm via convex programming, which can compute the exact locally-dense decomposition 
in real graphs with billions of edges, and proposes an $(1+\epsilon)$-approximation solution based on the Frank-Wolfe algorithm. 
Boob et al.~\cite{boob2020flowless} developed a simple iterative peeling algorithm Greedy++ to improve the quality of the subgraph over 
Charikar's greedy algorithm by drawing insights from the iterative approaches (multiplicative weights update) of convex optimization; 
the history of peeling information of nodes will help to escape the local solution to some extent. 
\cite{sawlani2020near} provided an algorithm for maintaining an $(1 - \epsilon)$-approximate (arbitrarily close to 1) densest subgraph 
within $O(\mathrm{poly}\log n)$ time over dynamic directed graphs, and extended to solve the problem on vertex-weighted static graphs. 
Feng et al.~\cite{feng2021specgreedy} proposed a generalized framework for addressing DSP and related 
problems~\cite{hooi2016fraudar, miyauchi2018finding, anagnostopoulos2020spectral, Tsourakakis2019NovelDS} and introduced SpecGreedy, 
an algorithm that leverages the graph spectral properties to a greedy peeling strategy to solve the generalized problem and speed up the detection. 
Chekuri et al.~\cite{chekuri2022densest} exploited the supermodular maximization and proposed more efficient $(1 - \epsilon)$-approximation algorithms in deterministic $\tilde{O}(m / \epsilon)$ time via approximate flow techniques for DSP, 
and gives evidence of the convergence and theoretical truthfulness of Greedy++, that is, 
it can converge to a $(1-\epsilon)$-approximation in $O(1/\epsilon^2)$ iterations; 
it also developed an $\frac{1}{2}$-approximation peeling algorithm for the densest-at-least-k subgraph. 
\cite{fazzone2022discovering} modified Greedy++ to have a quantitative certificate of the solution quality provided by the algorithm at each iteration. 
\cite{harb2022faster} proposed another iterative method using Proximal Gradient Method, which achieves $(1-\epsilon)$-approximation in $O(1/\epsilon)$, they also proposed a technique called Fractional Peeling to make use of the information in edge distribution. 
For directed graphs, the LP-based approach proposed by Charikar~\cite{charikar2000greedy} 
requires the computation of $n^2$ linear programs, and the $\frac{1}{2}$-approximation runs $O(n^3 + mn^2)$ time, 
\cite{khuller2009dense} provided more efficient implementations for these algorithms for undirected and directed graphs.

There is another research line that discovers the densest subgraph building upon some microstructures (motifs) in a graph, 
including the triangles~\cite{tsourakakis2015k, samusevich2016local}, cliques~\cite{sun2020kclist++, Tsourakakis2013Denser, fang2019efficient}, 
$k$-core~\cite{galimberti2017core}, $k$-club / $k$-plex, etc., and proposed corresponding different variants for the density measures. 
\cite{Mitzenmacher2015SLN, tsourakakis2015k} extended the DSP to the $k$-clique, and the $(p,q)$-biclique densest subgraph problems, 
which can be used to find large near-cliques. 
Tatti and Gionis~\cite{tatti2015density,tatti2019density} introduced the locally-dense graph decomposition method, 
which imposes certain insightful constraints on the $k$-core decomposition. 
\cite{fang2019efficient} proposed exact and approximate solutions by improving the flow-based exact algorithm by 
locating the densest subgraph in a specific $k$-core, which can be generalized 
by considering an arbitrary pattern graph and aiming to maximize the average number of occurrences of the pattern in the resulting subgraph. 
\cite{ma2020efficient} proposed $[x,y]$-core-based algorithms (both exact and approximation) 
with the divide-and-conquer strategy to find the densest subgraph for directed graphs.

When restrictions on the size of nodeset are imposed, the DSP also becomes NP-hard~\cite{andersen2009finding}, 
which is called the \textit{densest k-subgraph} (DkS). Its two variants called densest \textit{at-least-k subgraph} (DalkS) 
and densest \textit{at-most-k subgraph} (DamkS) are also NP-hard according to \cite{khuller2009finding,andersen2009finding}.

The dense subgraphs are used to detect communities~\cite{Chen2010Dense, costa2015milp, wong2018sdregion} and 
anomalies~\cite{prakash2010eigenspokes, beutel2013copycatch, hooi2016fraudar}. 
As one of the key characteristics, density, as well as other similar metrics like modularity~\cite{newman2006modularity}, 
associativity, and local density~\cite{qin2015locally}, are used as (part of) optimization objectives to detect community structures. 

%% file: 060conclusion.tex
In this paper, we propose \emph{\method} to detect the densest subgraph and prove it can converge to locally-dense decomposition. \emph{\method} redistributes edge weight in a locally optimal operation according to the linear programming of the DSP and quadratic programming of locally-dense decomposition. Besides, we develop a pruning technology using modified Counting Sort and prove that it is a subprocess of Greedy. We did a lot of experiments to exhibit its pruning efficiency on 26 real-world datasets and compare it with other algorithms about Greedy and Max-Flow. We also use it to prune the graph to speed up the iterative algorithms. In our experiments, \emph{\method} can converge to the optimal values both in the linear programming of the DSP and quadratic programming of locally-dense decomposition faster than other state-of-the-arts iterative algorithms, including Frank-Wolfe in \cite{danisch2017large}, Greedy++ in \cite{boob2020flowless} and FISTA in \cite{harb2022faster} and FW-M.

Through our study, there are far more interesting topics in DSP and LDD that we can study in future work. We list them as follows:\\
\vspace{-0.15in}
\begin{enumerate}[label={\arabic*.}]
    \item What is the relationship between the size of the graph after pruning and the size of the whole graph (it's possibly related to properties of Kronecker graphs in \cite{leskovec2010kronecker}, which concerns how graphs evolve)?
    \item In our experiments, Kamada-Kawai path-length has a close relationship with LDD because we use Kamada-Kawai path-length to decide the position of nodes. How can we theoretically analyze the relationship between these two things?
    \item Can we discover the relation between iteration count $T$ and the approximate ratio?
\end{enumerate}

These problems also show that there is a close relationship among DSP, locally-dense decomposition and other theories in the graph field like DkS and locally densest subgraph in \cite{ma2022finding}. Understanding the relationship among them can help us mine dense subgraphs in more efficient and meaningful ways, and we can use dense subgraphs to discover more characteristics in the graph.

%% file: 070ack.tex
We would like to express our sincere gratitude to Mr. Harb Elfarouk\cite{harb2022faster} for his invaluable contributions to this research. Mr. Elfarouk generously provided us with his algorithm and its codes, which formed the foundation of our work. His willingness to answer our questions and share his expertise played a crucial role in this study.

%% file: 080sup.tex
\begin{proof}[Proof of Corollary \ref{coroll:dks}]
     For DkS, we set $|\subnode|=k$, it has an upper bound as follows:
        \begin{equation*}
            \footnotesize
            \begin{aligned}
                \rho( \subnode) &=\frac{\sum_{e \in \edges(\subnode)} w_e}{|\subnode|}=\frac{\sum_{e=(u,v),\,u,v\in \subnode} f_e(u)+f_e(v)}{|\subnode|}\\
                &\le\frac{\sum_{e=(u,v),\,u,v\in \subnode}{(f_e(u)+f_e(v))}+\sum_{e=(u,v),\,u\in \subnode,v\not\in \subnode}{f_e(u)}}{|\subnode|}\\
                &=\frac{\sum_{u\in \subnode}{\sum_{e\ni u}{f_e(u)}}}{|\subnode|}=\frac{\sum_{u\in \subnode}{l_u}}{|\subnode|}\\
                &=\frac{\sum_{i=0}^{j-1}{\lambda_{i}*|B_i|}+(k-|B_{j-1}|)*\lambda_{j}}{k}\\
            \end{aligned}
        \end{equation*}
    Therefore, result 2 holds up. As for result 1, if $k=|B_j|$, according to property \ref{prop:one-way} the equivalency condition in line 2 holds up because $f_e(u)=0$ if $u\in\subnode$ and $v\not\in\subnode$, then result 1 holds up. These
    results also hold up for DalkS because if $|S|>k$, it will add more nodes into $\subnode$ with lower upper bounds on their loads.
\end{proof}

Before proofs of Lemma \ref{lem:k-core}, we introduce the definition of $k$-core from \cite{seidman1983network}: $k$-core is the maximal subgraph $G_k$ in graph G, the degree of where any vertex $v$ in $G_k$ is satisfied with $\setndeg{G_k}{v} \ge k$.

\begin{proof}[proof of Lemma \ref{lem:k-core}]
    The Greedy algorithm just moves any node whose degree is the lowest in the remaining graph $\mathcal{H}$. Let us remark $A(k)$ as the nodeset of k-core for specific k. We claim that when deleting a node $u\in A(k)$, there mustn't be any node $v\not\in A(k)$ in the remaining graph $\mathcal{H}$. We prove it by way of contradiction, w.l.o.g,  we set $u$ as the first node to be deleted in $A(k)$ in Greedy, therefore $u$ has the lowest degree in $\mathcal{H}$ and $\setndeg{\mathcal{H}}{v}\ge \setndeg{\mathcal{H}}{u}\ge \setndeg{A(k)}{u}\ge k$ for any node $v$ in $\mathcal{H}$, which produces a k-core subgraph with a larger size, and it leads to a contradiction.
\end{proof}

\begin{proof}[Proof of Theorem \ref{th:pruning}]
     We set $\mathcal{H}_k (k>0)$ is the remaining graph after k iterations in Algorithm \ref{alg:pruning} and $\mathcal{H}_0$ is the initial whole graph, $\mathcal{H}_{k+1}^{'}$ is the nodeset to be deleted in k+1 iteration. Therefore, $\mathcal{H}_{k+1}=\mathcal{H}_{k} \setminus \mathcal{H}_{k+1}^{'}$,
     
     \begin{equation*}
         \begin{aligned}
            \rho(\mathcal{H}_{k+1})
            &=\rho(\mathcal{H}_{k}\setminus \mathcal{H}_{k+1}^{'}) \\
            &=\frac{\weights(\edges(\mathcal{H}_{k}))-\weights(\edges(\mathcal{H}_{k+1}^{'}))}{|\mathcal{H}_{k}|-|\mathcal{H}_{k+1}^{'}|}\\
            &\ge\frac{\rho(\mathcal{H}_{k})\cdot |\mathcal{H}_{k}|-\sum_{v \in \mathcal{H}_{k+1}^{'}}{\setndeg{\mathcal{H}_k}{v}}}{|\mathcal{H}_{k}|-|\mathcal{H}_{k+1}^{'}|}\\
            &> \frac{\rho(\mathcal{H}_{k})\cdot |\mathcal{H}_{k}|-\sum_{v \in \mathcal{H}_{k+1}^{'}}{\rho(\mathcal{H}_{k})}}{|\mathcal{H}_{k}|-|\mathcal{H}_{k+1}^{'}|}\\
            &=\frac{\rho(\mathcal{H}_{k})\cdot |\mathcal{H}_{k}|-\rho(\mathcal{H}_{k})\cdot |\mathcal{H}_{k+1}^{'}|}{|\mathcal{H}_{k}|-|\mathcal{H}_{k+1}^{'}|}\\
            &=\rho(\mathcal{H}_{k})
         \end{aligned}
     \end{equation*}

    That means the density of graph $\mathcal{H}$ monotonically increases in iterations, then any deleted node has a lower degree (when it is being deleted) than the final density, i.e., $\delta$. Therefore, the remaining graph $\mathcal{H}$ is a $\delta$-core and it is the graph of some time of the greedy search according to Lemma \ref{lem:k-core}. 

    We claim that the process before getting the $\delta$-core is a monotonic increasing phase of density in Greedy. We can confirm two facts:
    \begin{enumerate}[label={\arabic*.}]
    \item During Greedy, if there is a node $u\in\mathcal{H}_{1}^{'}$ existing in the remaining graph, the deletion in Greedy will increase the density of the remaining graph. That's because if Greedy deletes a node $v\not\in \mathcal{H}_{1}^{'}$, then $\setndeg{\mathcal{H}}{v}\le \setndeg{\mathcal{H}}{u}<\rho(\graph)\le\rho(\mathcal{H})$. $\rho(\mathcal{H})$ will increase and $\rho(\graph)\le\rho(\mathcal{H})$ still holds up. If Greedy deletes the node $u$, now that $\setndeg{\mathcal{H}}{u}\le \rho(\mathcal{H})$, then $\rho(\mathcal{H})$ will also increase and $\rho(\graph)\le\rho(\mathcal{H})$ holds up.
    \item During Greedy, if there is a node $u\in\mathcal{H}_{k+1}^{'}$ existing in the remaining graph $\mathcal{H}>\mathcal{H}_k$, and there isn't any node belonging to $\mathcal{H}_{k}^{'}$. Then: $\rho(\mathcal{H})\ge\rho(\mathcal{H}_k)$ because we delete more nodes with lower degrees. Therefore, when we deletes a node $v\not\in\mathcal{H}_{k+1}^{'}$, then $\setndeg{\mathcal{H}}{v}\le \setndeg{\mathcal{H}}{u}<\rho(\graph)\le\rho(\mathcal{H}_{k})\le\rho(\mathcal{H})$, then $\rho(\mathcal{H})$ will increase and $\rho(\graph)\le\rho(\mathcal{H})$ holds up. If Greedy deletes the node $u$, now that $\setndeg{\mathcal{H}}{v}\le\rho(\mathcal{H})$, $\rho(\mathcal{H})$ will also increase and $\rho(\graph)\le\rho(\mathcal{H})$ holds up.
\end{enumerate}

Therefore, the density monotonically increases in Greedy before getting $\delta$-core.
\end{proof}

\begin{table*}[htbp]
    \centering
    \vspace{-0.1in}
    \caption{Dataset source and density of algorithms}
    \label{tab:density}
    \resizebox{0.98\linewidth}{!}{
    \begin{tabular}{l|c|c|c|c|c||c|c} \toprule
    \textbf{Dataset}&\textbf{Source}&\textbf{Type}&\textbf{Pruning}&\textbf{w\_app}&\textbf{exact}&\textbf{DLL}&\textbf{uw\_Pruning+DLL}\\
    \midrule
    ca-HepPh                    & Stanford’s SNAP database&scholar collaboration network
                       & 119     & 119     & 119     & 119     & 119        \\
    comm-EmailEnron             & Stanford’s SNAP database&communication
                      & 37.316  & 37.344  & 37.344  & 37.337  & 37.337     \\
    ca-AstroPh                  & Stanford’s SNAP database&scholar collaboration network
                       & 28.481  & 29.616  & 32.11   & 29.552  & 29.552     \\
    PP-Pathways                 & Stanford’s SNAP database&protein interaction network
& 74.159  & 77.995  & 77.995  & 77.995  & 77.995     \\
    soc-Twitter\_ICWSM          & konect&social network
                                         & 25.678  & 25.683  & 25.69   & 25.686  & 25.685    \\
    soc-sign\_slashdot          & Stanford’s SNAP database&social network
                      & 39.376  & 42.132  & 42.132  & 42.132  & 42.132     \\
    rating-StackOverflow        & konect& social network
                                         & 20.209  & 20.209  & 20.21   & 20.209  & 20.209     \\
    soc-sign\_epinion           & Stanford’s SNAP database&social network
                        & 80.168  & 85.599  & 85.637  & 85.589  & 85.589     \\
    ego-twitter       &  Stanford’s SNAP database&social network
                        & 59.281  & 68.414  & 69.622  & 68.414  & 68.414     \\
    soc-Youtube                 &  Stanford’s SNAP database&social network
                        & 45.545  & 45.58   & 45.599  & 45.576  & 45.577     \\
    comm-WikiTalk               &  Stanford’s SNAP database&communication
                       & 114.139 & 114.139 & 114.139 & 114.139 & 114.139    \\
    nov\_user\_msg\_time        & We own it privately.&social network
                           & 278.815 & 278.815 & 278.815 & 278.815 & 278.815    \\
    cit-Patents                 & AMiner scholar datasets&scholar collaboration network
                        & 132.776 & 135.706 & 137.261 & 135.706 & 135.706    \\
    soc-Twitter\_ASU            & ASU&social network
        & 593.847 & 593.847 & 593.847 & 593.847 & 593.847    \\
    soc-Livejournal             & Livejournal&social network
                                    & 104.596 & 104.601 & 104.609 & 104.603 & 104.603   \\
    soc-Orkut                   & Stanford’s SNAP database&social network
                    & 227.861 & 227.872 & 227.874 & 227.872 & 227.872   \\
    soc-SinaWeibo               & Network Repository&social network
                             & 164.967 & 165.193 & 165.415 & 165.196 & 165.191    \\
    \hline
    wang-tripadvisor            & konect&rating network
                                         & 13.442  & 13.873  & 14.082  &--         &--           \\
    rec-YelpUserBusiness        & Network Repository&rating network
                             & 87.825  & 87.912  & 87.921  &--         &--            \\
    bookcrossing                & konect &rating network
                                        & 92.148  & 92.322  & 92.374  &--         &--            \\
    librec-ciaodvd-review       & konect&rating network
                                         & 233.553 & 233.59  & 233.597 &--         &--            \\
    movielens-10m   & konect &rating network
                                        & 1351.35 & 1351.35 & 1351.35 &--         &--           \\
    epinions         & konect &rating network
                                         & 595.302 & 595.314 & 595.316 &--        &--           \\
    libimseti               & konect &social network
                                         & 1645.71 & 1645.73 & 1645.73 &--         &--       \\
    rec-movielens & Network Repository &rating network
                             & 1801.16 & 1801.16 & 1801.16 &--         &--            \\
    yahoo-song              & konect   &rating network
                                    & 46725.2 & 46725.2 & 46725.2 &--     &-- \\
    \hline
    \multicolumn{8}{l}{\multirow{2}{0.95\textwidth}{\textbf{note:} : \textbf{Pruning} (w\_Pruning,uw\_Pruning). \textbf{w\_app}: approximation algorihtms on weighted graph(Priority Tree,Pruning+Priority Tree,BBST). \textbf{exact}: exact algorithms(maxflow,w\_Pruning+maxflow). \textbf{DLL}:Doubly-linked list.}} \\
    \multicolumn{8}{l}{}   \\
    \bottomrule
    \end{tabular}
    }
\end{table*}

\begin{table*}[htbp]
    \centering
    \caption{Comparison between LOWD and baselines without the pruning.}
    \label{tab:withoutpruning}
    \vspace{-0.1in}
    \resizebox{0.95\textwidth}{!}{
        \begin{tabular}{l|l|l|l|l|l|l|l|l|l|l}\toprule
        &\multicolumn{5}{|c|}{Running time}&\multicolumn{5}{c}{The number of iteration rounds}\\
        \midrule
        \textbf{Dataset} & \textbf{LOWD} & \textbf{Greedy++} & \textbf{FW} & \textbf{FISTA} & \textbf{FW\_M} & \textbf{LOWD} & \textbf{Greedy++} & \textbf{FW} & \textbf{FISTA} & \textbf{FW\_M}\\
        \midrule
        ca-HepPh                 & 0.0168      & \first{0.0159}     & 0.0208        & 0.1109    & \second{0.0161}    & 1             & 1            & 2               & 3         & 1        \\
        comm-EmailEnron         & \second{0.1516}      & \first{0.0714}     & 0.5551        & 27.5346    & 0.3012    & 19            & 2            & 88              & 567       & 46       \\
        ca-AstroPh               & \first{0.3158}      & 0.8396     & 2.312         & 24.2470    & \second{0.6593}    & 36            & 33           & 343             & 375       & 92       \\
        PP-Pathways                    & \second{0.1174}      & \first{0.0333}     & 0.3214        & 12.6981    & 1.7086    & 8             & 1            & 18              & 131        & 167      \\
        soc-sign\_slashdot       & \second{0.2275}      & \first{0.0704}     & 0.7881        & 11.6623     & 1.1051    & 11            & 1            & 44              & 81       & 64       \\
        soc-sign\_epinion        & \first{0.3213}      & 1.4706     & 1.1853        & 11.1319    & \second{0.9581}    & 9             & 12           & 40              & 46        & 33       \\
        soc-Twitter\_ICWSM             & \first{1.8343}      & 30.6654    & 96.3405       & 1207.2601   & \second{6.1467}    & 49            & 114          & 2816            & 4901       & 175      \\
        rating-StackOverflow  & \second{2.1887}      & \first{1.0537}     & 10.0728       & 567.1578  & 18.7819   & 32            & 2            & 178             & 1125      & 341      \\
        ego-twitter                   & \first{1.6001}      & 2.9768     & 22.0971       & 2567.2520   & \second{1.7266}    & 34            & 21           & 550             & 8115       & 41       \\
        soc-Youtube             & \first{6.5253}      & 133.3086   & 82.3855       & 1001.4894  & \second{22.0866}   & 46            & 129          & 612             & 929      & 161      \\
        comm-WikiTalk            & \second{4.6758}      & \first{1.984}      & 118.307       & 1933.6041 & 414.9261  & 15            & 1            & 407             & 1151      & 1386     \\
        nov\_user\_msg\_time          & \first{341.1144}    & 4464.4998  & 2581.9086     &  9826.2465         & \second{718.8805}  & 97            & 174          & 624             &   454        & 177      \\
        cit-Patents\_AMINER           & \first{291.4452}    & 846.3628   & \second{837.8446}      & 27554.332 & 1523.2698 & 104           & 56           & 259             & 1913       & 471      \\
        soc-Twitter\_ASU              & \second{120.8906}    & \first{19.2912}    &    511.3170         &    12174.952       &   2047.9598       & 40            & 1            &     156            &   1244        &  603     \\
        soc-Livejournal & \makecell[c]{--}          & \makecell[c]{--}         & \makecell[c]{--}            & \makecell[c]{--}       & \makecell[c]{--}    & \makecell[c]{--}    & \makecell[c]{--}       & \makecell[c]{--}            & \makecell[c]{--}      & \makecell[c]{--}     \\
        soc-Orkut               & \makecell[c]{--}        & \makecell[c]{--}      & \makecell[c]{--}            & \makecell[c]{--}      & \makecell[c]{--}     & \makecell[c]{--}           & \makecell[c]{--}           & \makecell[c]{--}         & \makecell[c]{--}      & \makecell[c]{--}  \\
        soc-SinaWeibo\_NETREP          & \makecell[c]{--}       & \makecell[c]{--}      & \makecell[c]{--}      & \makecell[c]{--}    & \makecell[c]{--}   & \makecell[c]{--}    & \makecell[c]{--}     & \makecell[c]{--}    & \makecell[c]{--}     & \makecell[c]{--}   \\
        \hline
        wang-tripadvisor                    & \first{0.6686}      & 23.6761    & 30.209        & 198.1408  & \second{15.6705}   & 93            & 187          & 3894            & 4512  & 2007     \\
        rec-YelpUserBusiness  & \first{0.4621}      & 24.5879    & 5.1698        & 62.8514    & \second{2.1192}    & 54            & 228          & 626             & 886 & 239      \\
        bookcrossing        & \first{1.5177}      & 92.9406    & 27.4916       &910.8345& \second{8.877}     & 78            & 259          & 1398            &5510 & 448      \\
        librec-ciaodvd-review      & \first{4.2248}      & 75.7599    & 231.5822      &1882.2756& \second{38.6405}   & 79            & 143          & 4635            &4339& 773      \\
        movielens-10m               & \second{49.2341}     & \first{2.7029}     & 129.6344      &5329.5159& 60.719    & 159           & 1            & 450             &1611& 211      \\
        epinions                     & \first{30.4071}     &  3957.2229        & 4302.242      &$>$37786.731 & \second{743.6429}  & 64            &     688
       & 8893            &$>$10000& 1582     \\
        libimseti                           & \first{25.7101}     & 5276.5399        & 751.3750            &24832.745& \second{146.2627}        & 39            & 710     & 1196              &4632& 234       \\
        rec-movielens            & \first{50.4427}     & 6928.9143       & 306.0513       &24923.281 & \second{148.2539}       & 54            & 733          & 355         & 3195 & 172      \\
        yahoo-song                          & \makecell[c]{--}    & \makecell[c]{--} & \makecell[c]{--}        & \makecell[c]{--}      & \makecell[c]{--} & \makecell[c]{--}           & \makecell[c]{--}  & \makecell[c]{--}  & \makecell[c]{--}     & \makecell[c]{--} \\
        \hline
        \multicolumn{11}{l}{\multirow{2}{0.95\textwidth}{\textbf{note:} We ignore some datasets which are very large."$>$10000" and "$>$37786.731" means we run 10000 iterations(running time: 37786.731s) and can't still detect the densest subgraph.}} \\
        \multicolumn{11}{l}{} \\
        \bottomrule
        \end{tabular}
    }
\end{table*}

\begin{figure*}[t]
    \centering
    \subfigure[soc-Twitter\_ICWSM]{\includegraphics[width=0.24\linewidth]{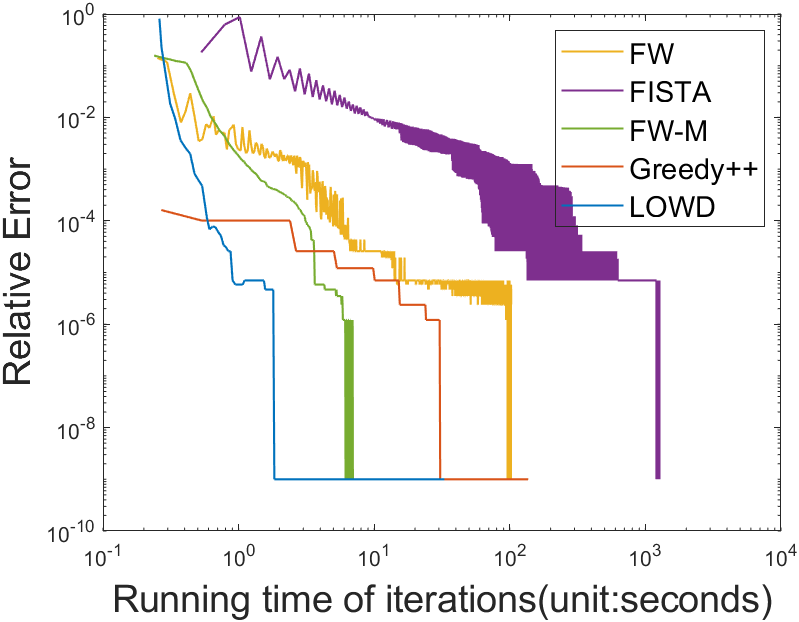}}
    \subfigure[soc-Youtube]{\includegraphics[width=0.24\linewidth]{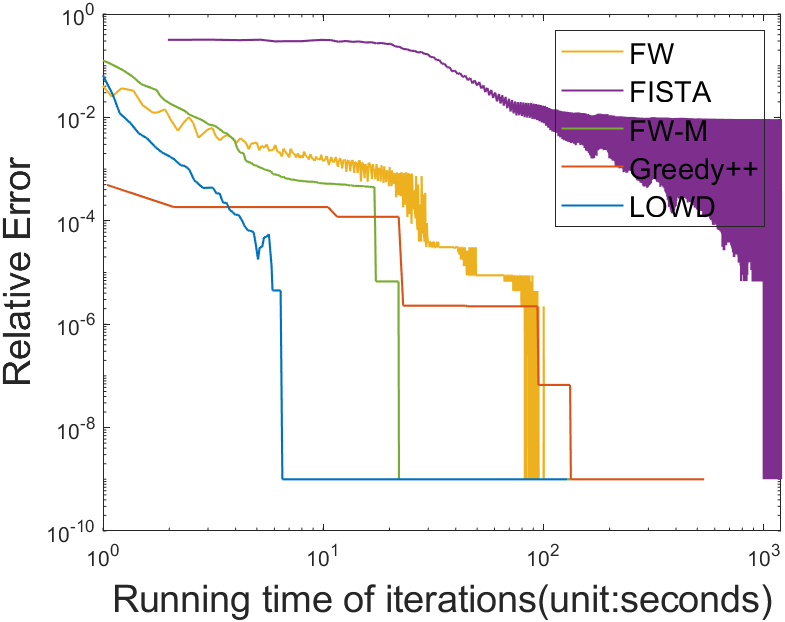}}
    \subfigure[comm-WikiTalk]{\includegraphics[width=0.24\linewidth]{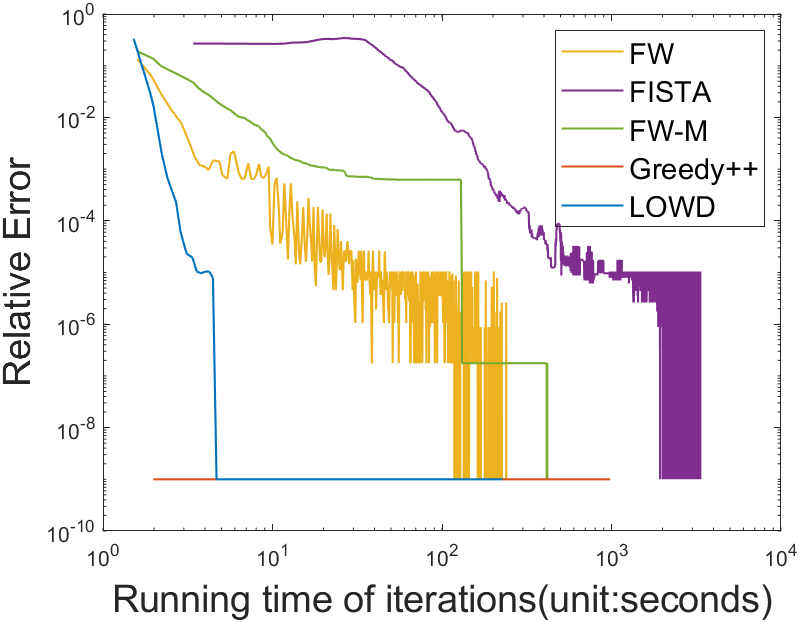}}
    \subfigure[nov\_user\_msg\_time]{\includegraphics[width=0.24\linewidth]{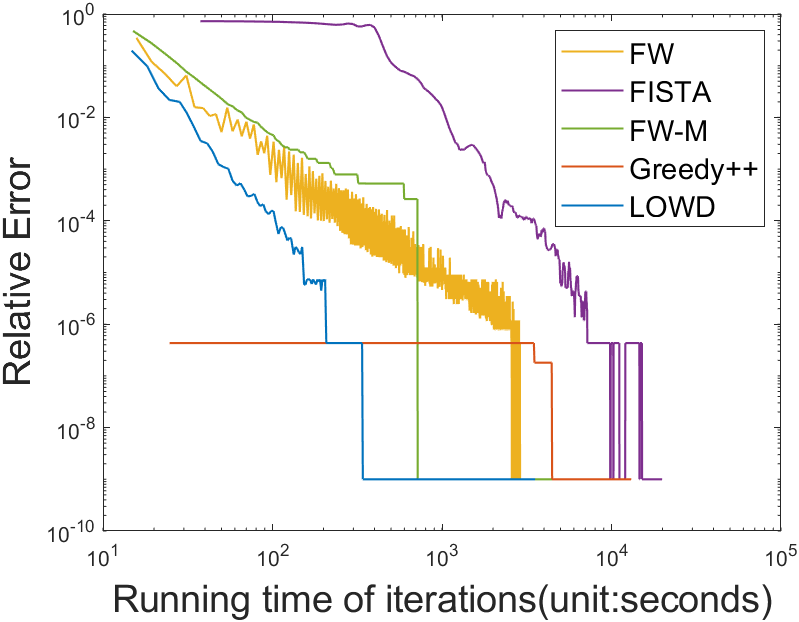}}
    \subfigure[cit-Patents\_AMINER]{\includegraphics[width=0.24\linewidth]{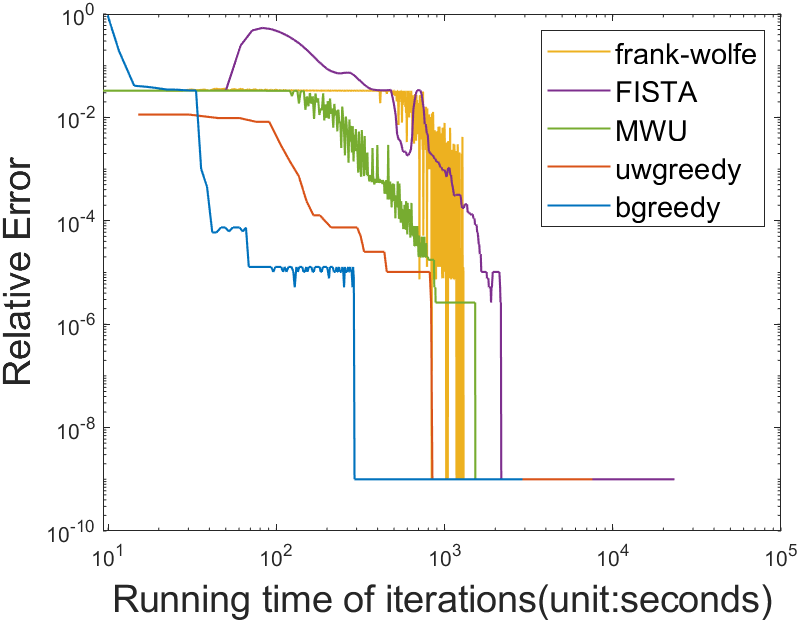}}
    \subfigure[rec-YelpUserBusiness]{\includegraphics[width=0.24\linewidth]{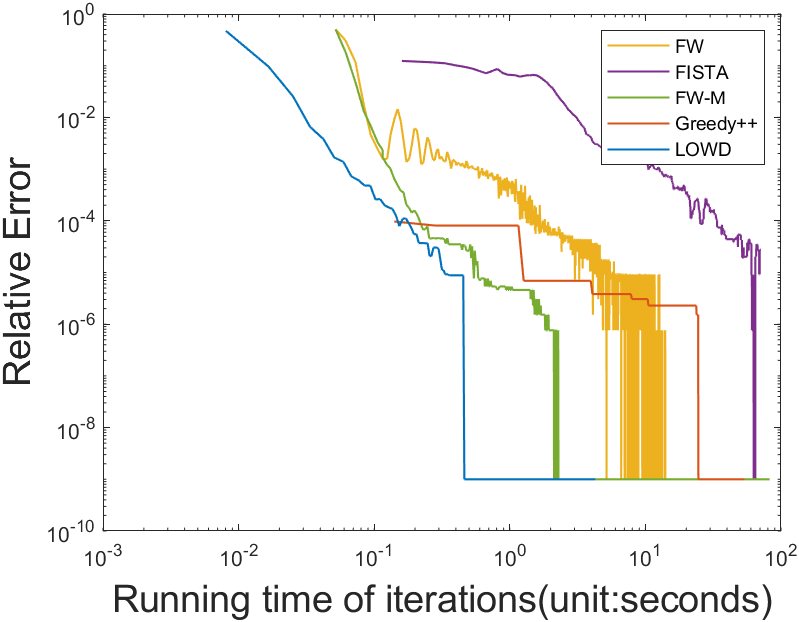}}
    \subfigure[librec-ciaodvd-review]{\includegraphics[width=0.24\linewidth]{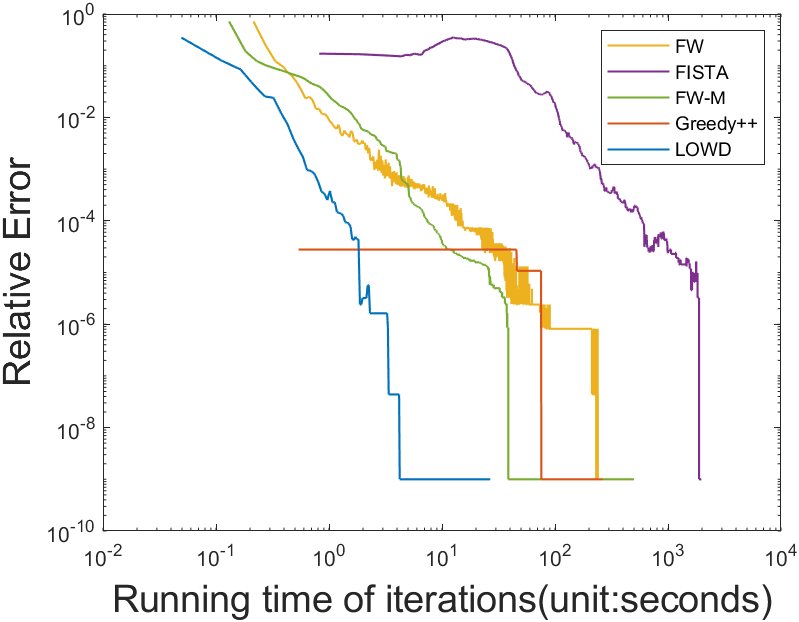}}
    \subfigure[movielens-10m]{\includegraphics[width=0.24\linewidth]{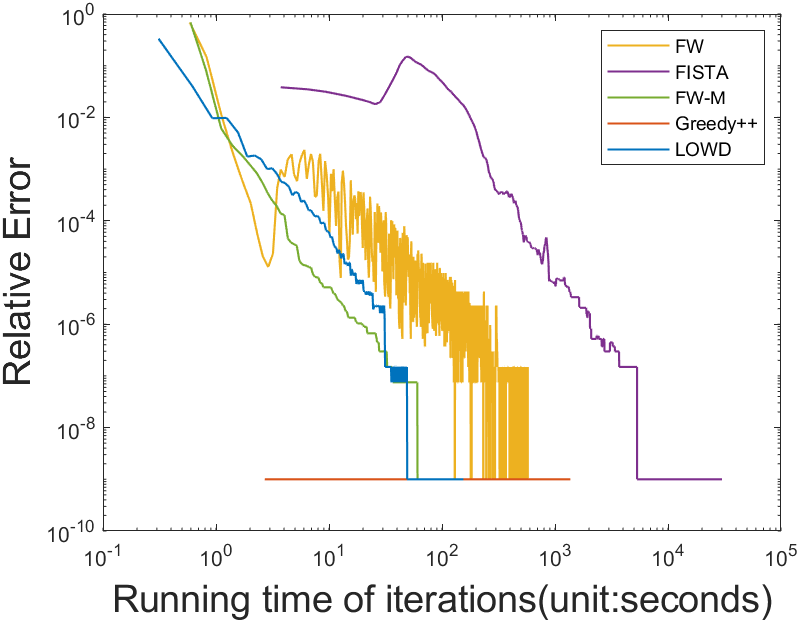}}
    \caption{densest}
    \label{fig:densest}
\end{figure*}

\begin{figure}[H]
    \centering
    \includegraphics[width=0.7\linewidth]{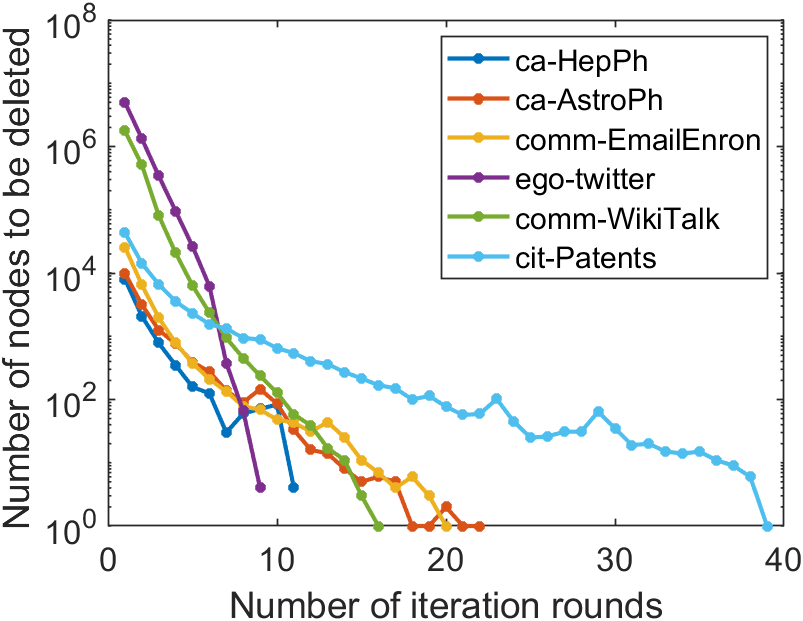}
    \vspace{-0.1in}
    \caption{Exponential decrease in the number of deleted nodes.}
    \label{fig:Exponential}
\end{figure}

\begin{algorithm}[htbp]
\caption{\textsc{Greedy DSPSolver}}
\label{alg:greedy}
\KwIn{Undirected graph $\graph$; density metric $\rho(\cdot)$}
\KwOut{$\optset$: the nodeset of the densest subgraph of $\graph$.}
   $\subnode, \, \optset \leftarrow \nodes$

   \While{$\subnode \ne \emptyset$}{
       \linecomment{find the vertex $u^{*}$ with the lowest degree in $\subnode$}
       
       $u^{*} \leftarrow \argmin_{u \in \subnode} \setndeg{\subnode}{u})$
       
       Remove $u^{*}$ and all its adjacent edges from $\graph$.

       \linecomment{$\subnode \backslash \{u\}$: the remaining nodeset without $u$}

       $\subnode \leftarrow \subnode \backslash \{u\}$

        \If{$\rho(\subnode) > \rho(\optset)$}{
        
            $\optset \leftarrow \subnode$
        }
       
   }

    \Return{$\optset$.}
\end{algorithm}

%% file: main.bbl
\begin{thebibliography}{100}
    \bibitem{shen2010spectral}
H.-W. Shen and X.-Q. Cheng, ``Spectral methods for the detection of network
  community structure: a comparative analysis,'' \emph{JSTAT}, 2010.

\bibitem{wong2018sdregion}
S.~W. Wong, C.~Pastrello, M.~Kotlyar, C.~Faloutsos, and I.~Jurisica,
  ``Sdregion: Fast spotting of changing communities in biological networks,''
  in \emph{SIGKDD'18}, 2018.

\bibitem{liu2019coupled}
Y.~Liu, L.~Zhu, P.~Szekely, A.~Galstyan, and D.~Koutra, ``Coupled clustering of
  time-series and networks,'' in \emph{SDM}.\hskip 1em plus 0.5em minus
  0.4em\relax SIAM, 2019.

\bibitem{cohen2003reachability}
E.~Cohen, E.~Halperin, H.~Kaplan, and U.~Zwick, ``Reachability and distance
  queries via 2-hop labels,'' \emph{SIAM Journal on Computing}, vol.~32, no.~5,
  pp. 1338--1355, 2003.

\bibitem{jin20093}
R.~Jin, Y.~Xiang, N.~Ruan, and D.~Fuhry, ``3-hop: a high-compression indexing
  scheme for reachability query,'' in \emph{Proceedings of the 2009 ACM SIGMOD
  International Conference on Management of data}, 2009, pp. 813--826.

\bibitem{goldberg1984finding}
A.~V. Goldberg, ``Finding a maximum density subgraph,'' 1984.

\bibitem{charikar2000greedy}
M.~Charikar, ``Greedy approximation algorithms for finding dense components in
  a graph,'' in \emph{APPROX'00}, 2000.

\bibitem{danisch2017large}
M.~Danisch, T.-H.~H. Chan, and M.~Sozio, ``Large scale density-friendly graph
  decomposition via convex programming,'' in \emph{WWW}, 2017.

\bibitem{sawlani2020near}
S.~Sawlani and J.~Wang, ``Near-optimal fully dynamic densest subgraph,'' in
  \emph{Proceedings of the 52nd Annual ACM SIGACT Symposium on Theory of
  Computing}, 2020, pp. 181--193.

\bibitem{boob2020flowless}
D.~Boob, Y.~Gao, R.~Peng, S.~Sawlani, C.~Tsourakakis, D.~Wang, and J.~Wang,
  ``Flowless: Extracting densest subgraphs without flow computations,'' in
  \emph{Proceedings of The Web Conference 2020}, 2020, pp. 573--583.

\bibitem{chekuri2022densest}
C.~Chekuri, K.~Quanrud, and M.~R. Torres, ``Densest subgraph: Supermodularity,
  iterative peeling, and flow,'' in \emph{Proceedings of the 2022 Annual
  ACM-SIAM Symposium on Discrete Algorithms (SODA)}.\hskip 1em plus 0.5em minus
  0.4em\relax SIAM, 2022, pp. 1531--1555.

\bibitem{harb2022faster}
E.~Harb, K.~Quanrud, and C.~Chekuri, ``Faster and scalable algorithms for
  densest subgraph and decomposition,'' \emph{Advances in Neural Information
  Processing Systems}, vol.~35, pp. 26\,966--26\,979, 2022.

\bibitem{khuller2009finding}
S.~Khuller and B.~Saha, ``On finding dense subgraphs,'' in \emph{Automata,
  Languages and Programming: 36th International Colloquium, ICALP 2009, Rhodes,
  Greece, July 5-12, 2009, Proceedings, Part I 36}.\hskip 1em plus 0.5em minus
  0.4em\relax Springer, 2009, pp. 597--608.

\bibitem{tatti2015density}
N.~Tatti and A.~Gionis, ``Density-friendly graph decomposition,'' in
  \emph{WWW}, 2015.

\bibitem{ma2022finding}
C.~Ma, R.~Cheng, L.~V. Lakshmanan, and X.~Han, ``Finding locally densest
  subgraphs: a convex programming approach,'' \emph{Proceedings of the VLDB
  Endowment}, vol.~15, no.~11, pp. 2719--2732, 2022.

\bibitem{tatti2019density}
N.~Tatti, ``Density-friendly graph decomposition,'' \emph{ACM Transactions on
  Knowledge Discovery from Data (TKDD)}, vol.~13, no.~5, pp. 1--29, 2019.

\bibitem{fang2019efficient}
Y.~{Fang}, K.~{Yu}, R.~{Cheng}, L.~V. {Lakshmanan}, and X.~{Lin}, ``Efficient
  algorithms for densest subgraph discovery,'' \emph{arXiv: Databases}, 2019.

\bibitem{bibby1974axiomatisations}
J.~Bibby, ``Axiomatisations of the average and a further generalisation of
  monotonic sequences,'' \emph{Glasgow Mathematical Journal}, vol.~15, no.~1,
  pp. 63--65, 1974.

\bibitem{jure2014snapnets}
J.~Leskovec and A.~Krevl, ``{SNAP Datasets}: {Stanford} large network dataset
  collection,'' \url{http://snap.stanford.edu/data}, Jun. 2014.

\bibitem{wan2019aminer}
H.~Wan, Y.~Zhang, J.~Zhang, and J.~Tang, ``Aminer: Search and mining of
  academic social networks,'' \emph{Data Intelligence}, 2019.

\bibitem{nr2015aaai}
R.~A. Rossi and N.~K. Ahmed, ``The network data repository with interactive
  graph analytics and visualization,'' in \emph{AAAI}, 2015. [Online].
  Available: \url{http://networkrepository.com}

\bibitem{ZafaraniLiu2009}
R.~Zafarani and H.~Liu, ``Social computing data repository at {ASU},'' 2009.
  [Online]. Available: \url{http://socialcomputing.asu.edu}

\bibitem{kunegis2013konect}
J.~Kunegis, ``Konect: the koblenz network collection,'' in \emph{WWW}, 2013,
  pp. 1343--1350.

\bibitem{hooi2016fraudar}
B.~Hooi, H.~A. Song, A.~Beutel, N.~Shah, K.~Shin, and C.~Faloutsos, ``Fraudar:
  Bounding graph fraud in the face of camouflage,'' in \emph{SIGKDD}, 2016.

\bibitem{veldt2021generalized}
N.~Veldt, A.~R. Benson, and J.~Kleinberg, ``The generalized mean densest
  subgraph problem,'' in \emph{Proceedings of the 27th ACM SIGKDD Conference on
  Knowledge Discovery \& Data Mining}, 2021, pp. 1604--1614.

\bibitem{lanciano2023survey}
T.~Lanciano, A.~Miyauchi, A.~Fazzone, and F.~Bonchi, ``A survey on the densest
  subgraph problem and its variants,'' \emph{arXiv preprint arXiv:2303.14467},
  2023.

\bibitem{luo2023survey}
W.~Luo, C.~Ma, Y.~Fang, and L.~V. Lakshman, ``A survey of densest subgraph
  discovery on large graphs,'' \emph{arXiv preprint arXiv:2306.07927}, 2023.

\bibitem{Gionis2015DSD}
A.~Gionis and C.~E. Tsourakakis, ``Dense subgraph discovery: Kdd 2015
  tutorial,'' in \emph{KDD}, 2015.

\bibitem{fang2022densest}
Y.~Fang, W.~Luo, and C.~Ma, ``Densest subgraph discovery on large graphs:
  applications, challenges, and techniques,'' \emph{Proceedings of the VLDB
  Endowment}, vol.~15, no.~12, pp. 3766--3769, 2022.

\bibitem{gallo1989fast}
G.~Gallo, M.~D. Grigoriadis, and R.~E. Tarjan, ``A fast parametric maximum flow
  algorithm and applications,'' \emph{SIAM Journal on Computing}, vol.~18,
  no.~1, pp. 30--55, 1989.

\bibitem{asahiro2000greedily}
Y.~Asahiro, K.~Iwama, H.~Tamaki, and T.~Tokuyama, ``Greedily finding a dense
  subgraph,'' \emph{Journal of Algorithms}, 2000.

\bibitem{feng2021specgreedy}
W.~Feng, S.~Liu, D.~Koutra, H.~Shen, and X.~Cheng, ``Specgreedy: unified dense
  subgraph detection,'' in \emph{Machine Learning and Knowledge Discovery in
  Databases: European Conference, ECML PKDD 2020, Ghent, Belgium, September
  14--18, 2020, Proceedings, Part I}.\hskip 1em plus 0.5em minus 0.4em\relax
  Springer, 2021, pp. 181--197.

\bibitem{miyauchi2018finding}
A.~Miyauchi and N.~Kakimura, ``Finding a dense subgraph with sparse cut,'' in
  \emph{CIKM}, 2018.

\bibitem{anagnostopoulos2020spectral}
A.~Anagnostopoulos, L.~Becchetti, A.~Fazzone, C.~Menghini, and
  C.~Schwiegelshohn, ``Spectral relaxations and fair densest subgraphs,'' in
  \emph{CIKM}, 2020.

\bibitem{Tsourakakis2019NovelDS}
C.~E. Tsourakakis, T.~Chen, N.~Kakimura, and J.~W. Pachocki, ``Novel dense
  subgraph discovery primitives: Risk aversion \& exclusion queries,''
  \emph{ECML-PKDD}, 2019.

\bibitem{fazzone2022discovering}
A.~Fazzone, T.~Lanciano, R.~Denni, C.~E. Tsourakakis, and F.~Bonchi,
  ``Discovering polarization niches via dense subgraphs with attractors and
  repulsers,'' \emph{Proceedings of the VLDB Endowment}, vol.~15, no.~13, pp.
  3883--3896, 2022.

\bibitem{khuller2009dense}
S.~Khuller and B.~Saha, ``On finding dense subgraphs,'' in \emph{Proceedings of
  the 36th International Colloquium on Automata, Languages and Programming:
  Part I}, ser. ICALP ’09.\hskip 1em plus 0.5em minus 0.4em\relax
  Springer-Verlag, 2009.

\bibitem{tsourakakis2015k}
C.~Tsourakakis, ``The k-clique densest subgraph problem,'' in \emph{WWW}, 2015.

\bibitem{samusevich2016local}
R.~Samusevich, M.~Danisch, and M.~Sozio, ``Local triangle-densest subgraphs,''
  in \emph{2016 IEEE/ACM International Conference on Advances in Social
  Networks Analysis and Mining (ASONAM)}.\hskip 1em plus 0.5em minus
  0.4em\relax IEEE, 2016, pp. 33--40.

\bibitem{sun2020kclist++}
B.~Sun, M.~Danisch, T.~Chan, and M.~Sozio, ``Kclist++: A simple algorithm for
  finding k-clique densest subgraphs in large graphs,'' \emph{Proceedings of
  the VLDB Endowment (PVLDB)}, 2020.

\bibitem{Tsourakakis2013Denser}
C.~Tsourakakis, F.~Bonchi, A.~Gionis, F.~Gullo, and M.~Tsiarli, ``Denser than
  the densest subgraph: extracting optimal quasi-cliques with quality
  guarantees,'' in \emph{SIGKDD}, 2013, pp. 104--112.

\bibitem{galimberti2017core}
E.~Galimberti, F.~Bonchi, and F.~Gullo, ``Core decomposition and densest
  subgraph in multilayer networks,'' in \emph{CIKM}, 2017.

\bibitem{Mitzenmacher2015SLN}
M.~Mitzenmacher, J.~Pachocki, R.~Peng, C.~Tsourakakis, and S.~C. Xu, ``Scalable
  large near-clique detection in large-scale networks via sampling,'' in
  \emph{KDD}, 2015.

\bibitem{ma2020efficient}
C.~Ma, Y.~Fang, R.~Cheng, L.~V. Lakshmanan, W.~Zhang, and X.~Lin, ``Efficient
  algorithms for densest subgraph discovery on large directed graphs,'' in
  \emph{SIGMOD}, 2020.

\bibitem{andersen2009finding}
R.~Andersen and K.~Chellapilla, ``Finding dense subgraphs with size bounds,''
  in \emph{WAW'09}.

\bibitem{Chen2010Dense}
J.~Chen and Y.~Saad, ``Dense subgraph extraction with application to community
  detection,'' \emph{IEEE TKDE}, 2010.

\bibitem{costa2015milp}
A.~Costa, ``Milp formulations for the modularity density maximization
  problem,'' \emph{European Journal of Operational Research}, 2015.

\bibitem{prakash2010eigenspokes}
B.~A. Prakash, A.~Sridharan, M.~Seshadri, S.~Machiraju, and C.~Faloutsos,
  ``Eigenspokes: Surprising patterns and scalable community chipping in large
  graphs,'' in \emph{PAKDD}, 2010.

\bibitem{beutel2013copycatch}
A.~Beutel, W.~Xu, V.~Guruswami, C.~Palow, and C.~Faloutsos, ``Copycatch:
  stopping group attacks by spotting lockstep behavior in social networks,'' in
  \emph{WWW}, 2013, pp. 119--130.

\bibitem{newman2006modularity}
M.~E. Newman, ``Modularity and community structure in networks,''
  \emph{Proceedings of the national academy of sciences}, 2006.

\bibitem{qin2015locally}
L.~Qin, R.-H. Li, L.~Chang, and C.~Zhang, ``Locally densest subgraph
  discovery,'' in \emph{SIGKDD}, 2015, pp. 965--974.

\bibitem{leskovec2010kronecker}
J.~Leskovec, D.~Chakrabarti, J.~Kleinberg, C.~Faloutsos, and Z.~Ghahramani,
  ``Kronecker graphs: an approach to modeling networks.'' \emph{Journal of
  Machine Learning Research}, vol.~11, no.~2, 2010.

\bibitem{seidman1983network}
S.~B. Seidman, ``Network structure and minimum degree,'' \emph{Social
  networks}, vol.~5, no.~3, pp. 269--287, 1983.

\bibitem{zhu2023fast}
Y.~{Zhu}, S.~{Liu}, W.~{Feng}, X.~{Cheng},``Fast Searching The Densest Subgraph And Decomposition With Local Optimality,'' \emph{arXiv preprint arXiv: 2307.15969}, 2023.
\end{thebibliography}
